\documentclass[journal]{IEEEtran}
\usepackage{amssymb}
\usepackage[english]{babel}
\usepackage{mathrsfs}
\newlength{\halfpagewidth}
\divide\halfpagewidth by 2

\newtheorem{corollary}{Corollary}

\newtheorem{theorem}{Theorem}
\newtheorem{lemma}{Lemma}
\newtheorem{remark}{Remark}

\newenvironment{proof}{{\indent \it {Proof}:\quad}}{\hfill $\square$\par}

\usepackage{amsmath}
\usepackage{graphicx}
\usepackage[colorlinks=true, allcolors=blue]{hyperref}
\usepackage{color}
\usepackage{lettrine}

\title{Cell-free Fluid Antenna Multiple Access Networks}
\author{Tianyu Han, Yongxu Zhu, Kai-Kit Wong,
Gan Zheng, and Hyundong Shin
\vspace{-6mm}

\thanks{T. Han and G. Zheng are with the School of Engineering, University of Warwick, Coventry, UK (e-mail: {tianyu.han, gan.zheng}@warwick.ac.uk).}
\thanks{Y. Zhu is with the National Communications Research Laboratory, Southeast University, Nanjing, China (e-mail: yongxu.zhu@seu.edu.cn).}
\thanks{K. Wong is affiliated with the Department of Electronic and Electrical Engineering, University College London, Torrington Place, UK and also with the Department of Electronic Engineering, Kyung Hee University, Yongin-si, Gyeonggi-do 17104, Republic of Korea (e-mail: kai-kit.wong@ucl.ac.uk)}
\thanks{H. Shin is with the Department of Electronics and Information Convergence Engineering, Kyung Hee University, Yongin-si, Gyeonggi-do 17104, Republic of Korea (e-mail: hshin@khu.ac.kr).}
\thanks{Corresponding author: Kai-Kit Wong, Yongxu Zhu}}

\begin{document}
\maketitle
	
\begin{abstract}
Fluid antenna enables position reconfigurability that gives transceiver access to a high-resolution spatial signal and the ability to avoid interference through the ups and downs of fading channels. Previous studies investigated this fluid antenna multiple access (FAMA) approach in a single-cell setup only. In this paper, we consider a cell-free network architecture in which users are associated with the nearest base stations (BSs) and all users share the same physical channel. Each BS has multiple fixed antennas that employ maximum ratio transmission (MRT) to beam to its associated users while each user relies on its fluid antenna system (FAS) on one radio frequency (RF) chain to overcome the inter-user interference. Our aim is to analyze the outage probability performance of such cell-free FAMA network when both large- and small-scale fading effects are considered. To do so, we derive the distribution of the received \textcolor{black}{magnitude} for a typical user and then the interference distribution under both fast and slow port switching techniques. The outage probability is finally obtained in integral form in each case. Numerical results demonstrate that in an interference-limited situation, although fast port switching is typically understood as the superior method for FAMA, slow port switching emerges as a more effective solution when there is a large antenna array at the BS. Moreover, it is revealed that FAS at each user can serve to greatly reduce the burden of BS in terms of both antenna costs and CSI estimation overhead, thereby enhancing the scalability of cell-free networks.
\end{abstract}
	
	
\section{Introduction}
\subsection{Motivation}
\IEEEPARstart{A}{n antenna} array at a base station (BS) has now become a default component for recent-generation wireless communications systems, with precoding essential to raise the spectral efficiency to the levels that are required by emerging applications \cite{Wong-mmimo2023,Zero_Forcing}. The fifth-generation (5G) currently uses $64$ antennas at the BS to perform codebook-based precoding to serve up to $12$ users on the same physical channel \cite{Villalonga-2022}. This was motivated by the massive multiple-input multiple-output (MIMO) concept, suggesting that interference can be averaged out in the extreme limit in the number of antennas \cite{Marzetta-2010}. In other words, a maximum ratio transmission (MRT) precoding would be enough to handle the inter-user interference \cite{MRT}. The strong interest in massive MIMO continues beyond 5G and the sixth generation (6G) is contemplating the deployment of extra-large MIMO (XL-MIMO) \cite{Wang-xlmimo}. However, the additional complexity of many more radio frequency (RF) chains and the explosive growth in power consumption are major obstacles \cite{Gustavsson-2021}, not to mention the overheads for channel estimation \cite{Yin-2013,He-2018}.
	
\subsection{FAS and Literature Review}
Motivated by the above, it is therefore important to explore ways that can reduce the burden of BS. An obvious choice is to strengthen the mobile users but mobile devices such as handset or tablet are normally limited in space and more constrained in terms of cost. An emerging reconfigurable antenna technology known as fluid antenna system (FAS) overcomes these issues by enabling shape and position flexibility in antenna \cite{Lee_fluid_Antenna}. FAS allows more diversity to be obtained in limited space and was first introduced to wireless communications by Wong {\em et al.}~in \cite{LCR,Fluid_antenna_system}. Since then, efforts have been made to ensure that the spatial correlation in the FAS channels is accurately captured \cite{wong2022closed_correlationn,khammassi_channel_model,block_correlation_model}. The diversity of a FAS receiver was analyzed in \cite{10130117} while \cite{10303274} even considered using FAS at both ends and studied the diversity-multiplexing tradeoff. Continuous FAS was also recently addressed in \cite{continuous_LCR}. In \cite{xu2023capacity}, it was found that FAS-equipped users could lift the uplink performance.

FAS can also combine with advanced systems such as non-orthogonal multiple access (NOMA) to improve communication performance \cite{FA_OMA_NOMA,OPtimization_OMA_NOMA}. Additionally, the use of FAS is not limited to mobile users and can be deployed at the BS to increase its degree-of-freedom (DoF) for various performance enhancement. In \cite{ISAC_FAS}, MIMO-FAS was adopted at the BS for integrated sensing and communication (ISAC). It is apparent that artificial intelligence (AI) methods will play an important role in the estimation and optimization of FAS channels \cite{Wang-aifas2024}. Channel state information (CSI) is essential to the operations of FAS and CSI estimation has recently been tackled by using methods such as sparse signal processing \cite{FA_TDMA,Dai-2023,10375559}.
	
It is important to understand that FAS is a conceptual idea not limited to any specific implementation mechanisms. It can be realized using a variety of methods such as reconfigurable pixels \cite{Rodrigo-2014,Jing-2022}, metamaterials \cite{Hoang-2021,Deng-2023}, and mechanically movable structures \cite{8060521}, etc. A recent article in \cite{new2024tutorial} provides a comprehensive tutorial on FAS covering a range of topics. Experimental results on FAS have also appeared in \cite{Shen-tap_submit2024} and \cite{Zhang-pFAS2024}. A branch of FAS, known as movable antennas, i.e., a FAS that employs mechanically movable antennas for position flexibility, also has drawn much attention \cite{zhu2024historical}. Recent works on this front tended to focus on the joint optimization that involves transmit power, precoding vector, and/or continuous antenna position that improves the rate performance under line-of-sight dominant channel conditions \cite{MA_Multiuser,MA_Multi_beam,MA_Near_Field}.

While it is apparent that FAS can be understood as a new DoF for enhancing the performance of wireless networks, our emphasis in this paper is on the possibility to improve network scalability for multiple access. Remarkably, FAS facilitates a new approach to multiple access, referred to as fluid antenna multiple access (FAMA) \cite{fast_FAMA}. With a high-fidelity signal in the spatial domain accessible by FAS, figuratively speaking, a mobile user can `ride' on the ups and downs of the fading channel and position itself to receive the signal under the most desirable condition. In \cite{fast_FAMA,Wong-ffama-comml}, it was proposed to change the position on a per-symbol basis, so that the data-dependent sum-interference plus noise signal is the weakest. This method is referred to as fast FAMA ($f$-FAMA). Later in \cite{slow_FAMA}, the slow FAMA ($s$-FAMA) scheme was studied, which was motivated by its simplicity in position (or port) switching. In $s$-FAMA, the port of FAS is only switched when the channel, not data, changes. The beauty of FAMA is that under the condition of rich scattering, precoding at the BS is not required and each user handles its interference on its own independently on one RF chain, and without interference cancellation. For line-of-sight (LoS) dominant channel conditions, the performance of FAMA can evidently be very different. In \cite[Section V-E]{new2024tutorial}, it has been demonstrated that FAMA can still be very effective when combining with simple precoding at the BS under LoS-dominant channel conditions. In particular, it is possible to reduce the CSI requirement of the BS to knowing statistical CSI only, yet still enabling effective spatial multiplexing when the users employ FAMA techniques. Moreover, opportunistic scheduling was reported to be effective in greatly enhancing the interference immunity of FAMA in \cite{wong2023opportunistic}. Also, a distributed team-inspired deep reinforcement learning approach was proposed in \cite{Waqar-2024} to realize opportunistic FAMA under dynamic environments where users can self-optimize their decisions.

The central idea in FAMA is to select the `best' receiving port to avoid interference, which requires the CSI of all the ports. Current approaches attempt to reconstruct the CSI across all the ports from the CSI of a few estimated ports \cite{Port_Selection}. By contrast, it is a lot more challenging to identify the best port in $f$-FAMA which requires the knowledge of the instantaneous signal-to-interference ratio (SIR) at all the ports, a problem that was studied in \cite{Wong-ffama-comml}. Moreover, machine learning methods solved the CSI estimation problem for $s$-FAMA in \cite{DL_s_FAMA,cGAN_s_FAMA}. Later, \cite{Generalized_CSI} proposed a generalized CSI estimation approach by using an asymmetric graph-masked autoencoder. More in-depth discussion on FAMA and its CSI requirements is given in \cite{new2024tutorial} and a recent survey for FAMA appears in \cite{Rabie-2024}.

Despite anticipating promising results in FAS and FAMA, a major difficulty in the required performance analysis is to deal with the spatial correlation among the ports, which makes it extremely challenging to derive the joint probability density function (PDF) and cumulative density function (CDF) for both the desired signal and interference signal. Even if the joint PDF and CDF are available, performance evaluation typically involves layers of integration, disallowing useful insights to be drawn. Fortunately, approximate models can be employed to permit closed-form analytical results, as achieved in, e.g., \cite{fast_FAMA,slow_FAMA}, and there is also an enhancement technique that ensures the accuracy of approximate models \cite{block_correlation_model}.

Before ending this literature review, it is worth clarifying the difference of FAMA from the notion of conventional antenna selection. Historically, antenna selection refers to the situation where fixed antennas with sufficient spacing between them are deployed and the best one is selected for reception. Normally, at least a half-wavelength distance is expected between two adjacent antennas to ensure signal independence for diversity benefits. Closely spaced antennas are never advised as correlation is deemed undesirable. This is the reason why antenna selection was never considered for multiple access because it would have required an extremely large number of antennas (and hence a gigantic space, recalling the requirement of half-wavelength antenna spacing) to be possible. On the contrary, FAMA thrives even under spatially correlated channels and it requires only {\em one} antenna that is able to access the received signals in a prescribed space. In FAMA, it is the multiuser diversity (i.e., signal independence {\em between users}) which is at play here, and the spacing between the FAS ports at a given user plays a different role in this objective. Specifically, FAMA performs better if the spatial resolution of FAS is higher for a given size, meaning that correlation helps. In other words, a clear distinction is that FAS manages to effectively utilize the spaces left between the fixed antennas in the conventional antenna selection system that has been overlooked before. This also comes as a timely revolution when such reconfigurable antennas for position flexibility become available.

\subsection{Aim and Contributions}
The emergence of FAMA offers an important alternative to the multiuser MIMO approach when CSI at the transmitter side is impossible, as is typical in device-to-device communication scenarios. That being said, multiuser MIMO is the pinnacle of the physical layer in cellular networks. In this case, the BS is required to possess the CSI of the user it is serving but for complexity reasons, it is preferred that the BS does not know the CSI of other users who may experience interference from its downlink transmission. Consequently, MRT precoding is a more scalable solution than methods such as zero-forcing that require full knowledge of the CSI of all other users. It would be interesting to find out if the use of FAS at each user can ease the burden of BS. Despite the valuable insights offered by previous work on FAS and FAMA, it is not at all understood how well FAMA could work with the multiuser MIMO setup in cellular networks, which is the goal of this paper. 

While our interest is to analyze the performance of FAMA with multiuser MIMO in cellular networks, we recognize that a cell-free setup is increasingly relevant due to its high spectral efficiency \cite{Ngo-2017}. As a result, our study considers the multi-cell model in which the `cell' structure is only here to specify the coverage of a BS and its association to the user, and all users share the same time-frequency physical channel. Each BS has multiple fixed antennas using MRT precoding to transmit to one serving user which is equipped with FAS utilizing FAMA to help alleviate the inter-user interference. According to the terminology in \cite{new2024tutorial}, this is a Rx-MISO-FAS model in the cell-free architecture. Different from the previous work, our work considers the use of MRT in the downlink together with FAMA and that both large- and small-scale fading are considered. 

Our objective is to substantiate the synergy between MRT precoding at the BSs and FAMA at the users in cell-free\footnote{In theory, a cell-free network will use all the antennas at all the BSs to serve a user. However, in practice, due to complexity reasons, only a few BSs can perform coordinated beamforming to serve a shared user. In our model, a BS performing MRT may be interpreted as a few coordinated BSs.} networks. More precisely, in a traditional cell-free network, to manage inter-user interference effectively, all BSs are required to obtain the CSI for all users. This comprehensive CSI knowledge enables the use of interference mitigation techniques such as zero-forcing. Unfortunately, the high CSI requirements and the need for extensive cooperation among BSs impose significant constraints on scalability, especially in large-scale deployments. We aim to investigate whether the use of FAS on the user side can alleviate the CSI burden on BSs, thereby enhancing the scalability of the cell-free network structure.
	
In summary, we have made the following contributions:
\begin{itemize}
\item Considering both Rayleigh fading and large-scale fading for cell-free FAMA networks with MRT precoding, we derive the exact joint distribution of the desired signal and interference at the user in the case for $f$-FAMA, while an approximate distribution of the interference is derived in the case of $s$-FAMA. Our derivation method leverages the correlation among ports to directly capture the joint PDF, without relying on the conditional PDF to ensure independence. This approach bypasses the generalized chi-squared distribution, which is very difficult to handle, allowing both the joint PDF and CDF to be derived, and making the performance analysis more tractable.
\item Using the distributions, we obtain two integral-form expressions for the outage probability of MRT-$f$-FAMA and MRT-$s$-FAMA, respectively, and provide an alternative form for each to facilitate numerical evaluation.
\item Our simulation results reveal that with a sufficient number of fixed antennas at the BS for MRT precoding, $s$-FAMA can outperform $f$-FAMA, which is not known possible before. This finding shows that $s$-FAMA is more desirable than $f$-FAMA not only because of its practical simplicity but also its superior performance when MRT precoding is used. Additionally, users equipped with FAS can greatly improve outage probability performance, and effectively reduce the burden of BS in its number of antennas.
\end{itemize}

The remainder of this paper is organized as follows. Section \ref{sec:model} introduces the FAMA network model with MRT precoding in the multi-cell setup. Then Section \ref{sec:analysis} presents the main analytical results and obtains the outage probability expressions. In Section \ref{sec:results}, numerical results are provided and finally, we provide some concluding remarks in Section \ref{sec:conclude}.

For the rest of this paper, the PDF for Gamma and Nakagami distributions will be used frequently. Thus, we find it useful to provide them below:
\begin{align}
{\rm Gamma}(\tau;\alpha,\beta) &= \frac{\tau^{\alpha-1}e^{-\beta \tau}\beta^{\alpha}}{\Gamma(\alpha)},\label{Gamma_PDF}\\
{\rm Nakagami}(\tau;\alpha,\beta) &= \frac{2{\alpha}^{\alpha}\tau^{2\alpha-1}e^{-\frac{\alpha \tau^2}{\beta}}}{\Gamma(\alpha)\beta^{\alpha}},\label{Nakagami_PDF}
\end{align}
where $\Gamma(\cdot)$ denotes the gamma function, and $\alpha$ and $\beta$ are the corresponding parameters while $\tau$ is the random variable.

\section{System Model}\label{sec:model}
We consider a downlink cell-free network with $U+1$ BSs, each serving a corresponding cell centered around it. Each BS is equipped with $N$ fixed antennas. Users are equipped with a single fluid antenna, served by the nearest BS. Considering a time-division multiple-access (TDMA) setup, within each time slot, each BS supports only one user within its coverage area, and all the BSs share the same time-frequency physical channel. Hence, each user is subjected to interference coming from the other $U$ BSs. Moreover, each BS is assumed to have the CSI to perform MRT precoding to serve its associated user. The concept behind our proposed system is to employ MRT beamforming at each BS to enhance signal power, while interference is mitigated at the receiver using FAS.\footnote{Note that the consideration of MRT precoding in this work is not because it is rate optimal but because it is more scalable than other methods such as zero-forcing in terms of the CSI requirement at the BS. This is particularly relevant in the cell-free network structure considered in this work. If zero-forcing were to be used in our model, it would have required each BS to know the CSI of all the other users in the entire network, which is practically not possible. Therefore, in our work, the more feasible way to handle the inter-user interference is by the FAMA techniques at the FAS-equipped users. The MRT precoding at the BS on the other hand is used to enhance the channel gain of the desired user, ignoring the interference it is going to cause to other users and leaving it to be dealt with FAS at each user.}

At each user, the FAS can be switched instantly to one of the $K$ preset locations (a.k.a.~port) evenly distributed along a linear dimension of length, $W\lambda$, where $W$ is a scaling factor and $\lambda$ is the carrier wavelength. Delay-free port switching can be easily achieved by adopting the reconfigurable pixels \cite{Zhang-pFAS2024} and metamaterial approaches \cite{Deng-2023}. Additionally, the results of this paper can be easily extended to the case with a two-dimensional FAS surface at each user but we opt for the linear FAS to simplify our discussion. Fig.~\ref{fig: MISO_FAS_Network} depicts our model. 

\begin{figure*}[]
\centering
\includegraphics[width=0.7\linewidth]{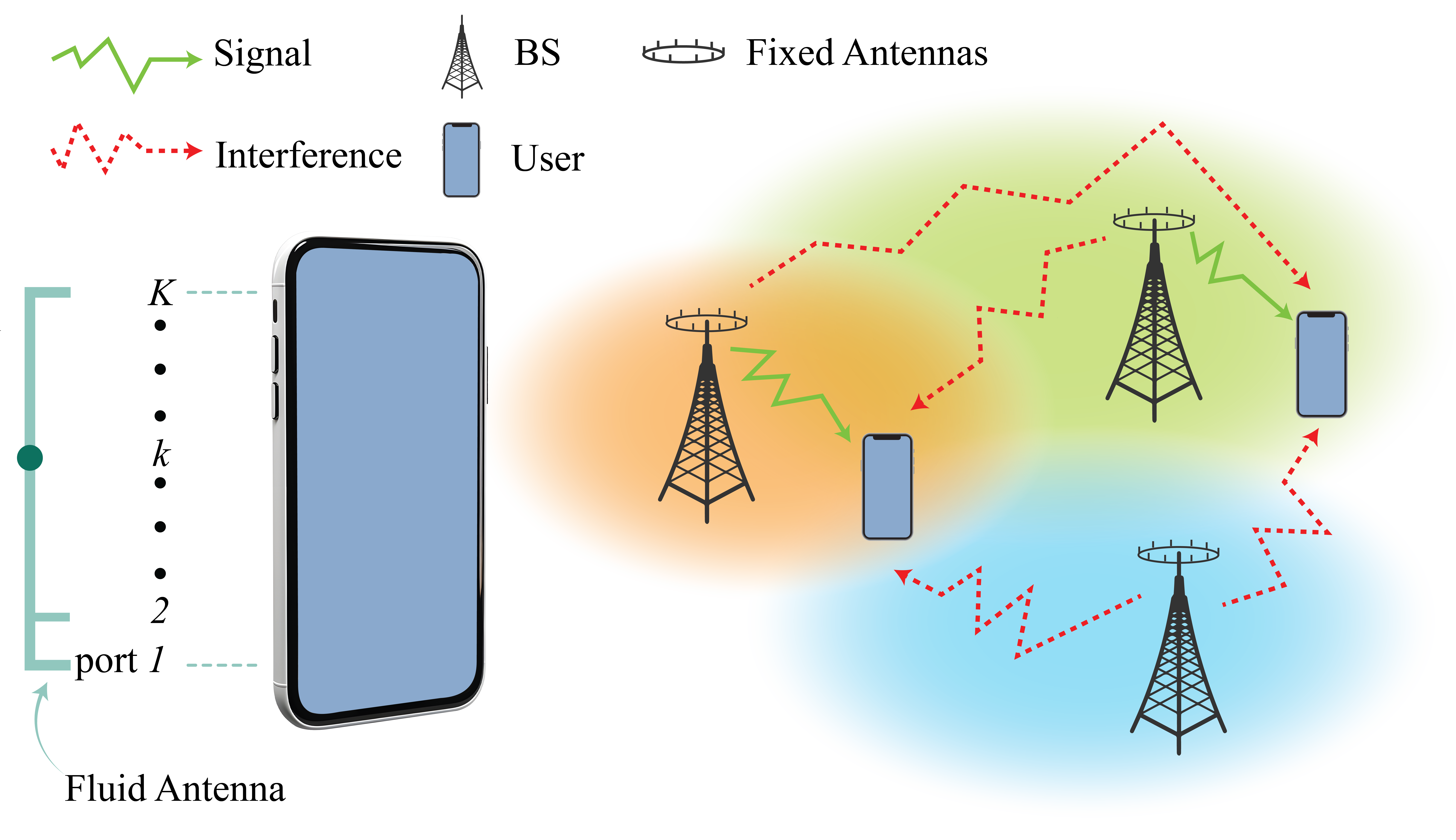}
\caption{The concept of a cell-free FAMA network (i.e., a multi-cell Rx-MISO-FAS system) where each user equipped with a single fluid antenna is served by its nearest BS and interference comes from the other BS transmitters.}\label{fig: MISO_FAS_Network}
\end{figure*} 

\subsection{Channel Model}
To describe our model, we consider an average user with the user index omitted for notational convenience. For small-scale fading between the $n$-th fixed antenna of the nearest BS and the $k$-th port of the fluid antenna at this target user, the complex channel can be modelled as \cite{wong2022closed_correlationn,rician_rv}
\begin{multline}\label{correlated_channel}
h_{n,k}=\sigma \left(\sqrt{1-\mu^2} x_{n,k} + \mu x_{n,0} \right)\\
+j\sigma \left(\sqrt{1-\mu^2} y_{n,k} + \mu y_{n,0}\right),
\end{multline}
where $x_{n,0},x_{n,1},\dots, x_{n,k}$ and $y_{n,0},y_{n,1},\dots, y_{n,k}$ are independent real-valued Gaussian random variables with zero mean and variance of $\frac{1}{2}$. Consequently, $h_{n,k}$ follows a complex normal distribution as $\mathcal{CN}(0, \sigma^2)$ while the magnitude $\left | h_{n,k} \right |$ follows a Rayleigh distribution \cite{rician_rv}. Further, we can define $h_{n,0}\triangleq \sigma (x_{n,0}+j y_{n,0})$ as the common random variable that links the random channels at the ports with the correlation parameter $\mu$ that can be chosen as \cite{wong2022closed_correlationn}
\begin{align} \label{Correlation_Parameter}
\mu ^2={\left|\frac{2}{K(K-1)}\sum _{k=1}^{K-1}(K-k)J_0{\left(\frac{2\pi kW}{K-1}\right)}\right|}, 
\end{align}
where $J_0(\cdot)$ is the zero-order Bessel function of the first kind. 

Moreover, we also consider the large-scale channel effects in terms of the path loss. With MRT precoding at the BS and assuming the distance $r_0$, the overall channel becomes
\begin{align}\label{MRT_signal_port_k}
g_{k} &= r_{0}^{-\alpha}\frac{\mathbf{h}_k\mathbf{h}_{k}^{\dagger}}{\left\| \mathbf{h}_k\right\|}\nonumber \\ 
&=\sqrt{\sum_{n=1}^{N}r_{0}^{-\alpha}\left |  h_{n,k} \right |^2}, 
\end{align}
where $ \mathbf{h}_k = \left[h_{1,k}, h_{2,k}, \dots, h_{N,k}\right]$, $(\cdot)^{\dagger}$ denotes the conjugate transpose, $\|\cdot\|$ represents the Euclidean norm, and $\alpha$ denotes the path loss exponent. Note that MRT is aimed at countering the small-scale fading only and serves to co-phase the signals arriving at the user's ports. For clarity, we also find it useful to define $g_0\triangleq\sqrt{\sum_{n=1}^{N}r_{0}^{-\alpha}|h_{n,0}|^2}$.

Then following \cite[Remark 1]{MRT_interference} and \cite[(29) \& (30)]{MRC_interference}, the interference channel at the $k$-th port of the user of interest with MRT precoding can be modeled as
\begin{align}\label{interference_channel}
g_{k}^{(i)} =  r_{i}^{-\frac{\alpha}{2}}h_{k}^{(i)},
\end{align} 
in which $r_{i}$ denotes the distance between the $i$-th interfering BS and the user, and $h_{k}^{(i)}$ is the small-scale fading following a complex Gaussian distribution as $\mathcal{CN}(0, \sigma^2)$. For conciseness, we define the distance vector as ${\bold r} = [r_{0}, r_{i}, \dots, r_{U}]$ with its elements sorted in an increasing order.

\subsection{Signal Model with $f$-FAMA}
Under this cell-free FAMA model, the received signal at the $k$-th FAS port of the user of interest is given by
\begin{align}\label{signal_with_interference}
z_{k} = sg_{k} +  \sum_{i=1}^{U}s_{i}g_{k}^{(i)} + \eta_{k} \equiv sg_{k} +g_{k}^{\rm I}+ \eta_{k},
\end{align}
where $s$ denotes the information-bearing symbol for the desired user, $s_i$ is the symbol intended for another user from the $i$-th interfering BS, and $\eta_k$ represents the zero-mean complex additive white Gaussian noise with variance of $\sigma_\eta^2$. The power of the information symbols from all the BSs is assumed to be the same and ${\rm E}[|s]^2]={\rm E}[|s_{i}|^{2}] = \sigma_{s}^2$, $\forall i$. 

In this paper, our focus is on the interference-limited scenario where $\sigma_s^2\gg \sigma_\eta^2$. In the sequel, we will ignore the noise term when developing our analytical results. To avoid the inter-user interference, it was proposed in \cite{fast_FAMA} that the user finds the ratio of the instantaneous desired signal to the interference signal at each port $k$, i.e.,
\begin{align}\label{eq:sir-fast}
\mathrm{SIR}_{k}^{{\rm [f]}}=\frac{\left|g_{k}\right|^2|s|^2}{\left|g_{k}^{\rm I}\right|^2},
\end{align}
where $\left|g_{k}^{\rm I}\right|^2$ denotes the instantaneous (data-dependent) sum-interference energy. For convenience, we call (\ref{eq:sir-fast}) the {\em instantaneous} signal-to-interference ratio (SIR).

Specifically, the user can select the port that maximizes the instantaneous SIR so that 
\begin{equation}\label{eq:c-ffana}
\left|g_{f\text{-}{\rm FAMA}}\right|=\max_k \frac{\left|g_k\right|}{\left|g_k^{\rm I}\right|},
\end{equation}
in which the symbol $s$ for the user concerned is not included in the maximization because it is not dependent upon $k$. For techniques to estimate the ratios in the above maximization, readers are referred to \cite{Wong-ffama-comml}. The expression \eqref{eq:c-ffana} indicates that the port maximizing the instantaneous SIR among $K$ ports is selected, an approach referred to as $f$-FAMA. The overall network performance can be understood by considering the outage probability
\begin{align}\label{OUT_f_FAMA}
{\mathcal O}_{{\rm MRT}\text{-}f\text{-}{\rm FAMA}} &= \mathbb{P}\left\{\sigma_{s}^2\left|g_{f\text{-}{\rm FAMA}}\right|^2 <\gamma_{\rm th}^{{\rm [f]}}\right\} \notag\\
&= \mathbb{P}\left\{\left|g_{f\text{-}{\rm FAMA}}\right|<\Theta_{\rm f}\right\},
\end{align}
where $\Theta_{\rm f}\triangleq \sqrt{\frac{\gamma_{\rm th}^{{\rm [f]}}}{\sigma_{s}^2}}$ and $\gamma_{\rm th}^{{\rm [f]}}$ is the prescribed SIR threshold. Note that the above outage probability formulation treats (\ref{eq:c-ffana}) as the random resulting channel after MRT precoding and $f$-FAMA port selection and then computes the SIR by assuming $|s|^2=\sigma_s^2$ before comparing it with the SIR threshold, $\gamma_{\rm th}^{{\rm [f]}}$. This is valid for constant-modulus modulation schemes.

\subsection{Signal Model with $s$-FAMA}
While \cite{Wong-ffama-comml} shows that the ratio (\ref{eq:sir-fast}) can actually be estimated on a per-symbol basis for the maximization (\ref{eq:c-ffana}), the estimation accuracy needs to be much improved or the performance of $f$-FAMA could degrade a lot. With the signal model in (\ref{signal_with_interference}), it is possible to select the port based on a simpler criterion. In \cite{slow_FAMA}, the {\em average} SIR is considered instead, given by
\begin{align}\label{SIR_SLow}
\mathrm{SIR}_{k}^{\rm [s]} =\frac{{\rm E}\left[\left|s g_k\right|^2\right]}{{\rm E}\left[\left|g_{k}^{\rm I}\right|^2\right]}=\frac{\left|g_k\right|^2}{\left|g_{k}^{\rm [s]} \right|^2},
\end{align}
where
\begin{align}
\left|g_{k}^{\rm [s]} \right|^2 = \sum_{i=1}^{U}\left |r_{i}^{-\frac{\alpha}{2}} h_{k}^{(i)} \right |^2.
\end{align}
Note that (\ref{SIR_SLow}) only depends on the instantaneous CSI but not the data. For slow fading channels, it is practically achievable. 

In this case, port selection can be done by
\begin{align}
\left|g_{s\text{-}{\rm FAMA}}\right|=\max_k \frac{\left|g_k\right|}{\left|g_k^{\rm [s]} \right|}.
\end{align}
As such, in the $s$-FAMA approach \cite{slow_FAMA}, the user only needs to switch the port once during each channel coherence time. The outage probability for this network is found as
\begin{align}\label{OUT_s_FAMA}
{\mathcal O}_{{\rm MRT}\text{-}s\text{-}{\rm FAMA}}=\mathbb{P}\left\{\left|g_{s\text{-}{\rm FAMA}}\right| < \Theta_{\rm s} \right\},
\end{align}
where $\Theta_{\rm s} \triangleq\sqrt{\gamma_{\rm th}^{\rm [s]}}$ and $\gamma_{\rm th}^{\rm [s]}$ is the SIR threshold.

\section{Main Results}\label{sec:analysis}
In this section, we present our principal results that evaluate the performance of both MRT-$f$-FAMA and MRT-$s$-FAMA. We first represent the distributions of the received signal and the interference with fast and slow port switching, respectively. Afterwards, we provide integral-form expressions for the outage probability for both systems, along with an alternative form for numerical computation for each. Our analysis focuses on an average user so it represents the average performance of any user in the FAMA network using MRT precoding.

\subsection{Signal Distribution at FAS}
The following theorem characterizes the distribution of the received signal after MRT precoding at the $k$-th port of the fluid antenna and the correlation among the ports.

\begin{theorem}\label{g_k_Nakagami_RV}
The PDF of the precoded channel of the desired signal at the $k$-th port, $\left|g_{k}\right|$, is given by
\begin{align} \label{PDF_gk_Nakagami}
f_{\left|g_{k}\right|}(\tau_{k}) = {\rm Nakagami}(\tau_{k}; \omega, \iota),
\end{align}
where 
\begin{equation}\label{eq:omega}
\left\{\begin{aligned}
\omega &= N,\\
\iota&= r_{0}^{-\alpha}N\sigma^2,
\end{aligned}\right.
\end{equation}
and ${\rm Nakagami}(\cdot,\cdot)$ represents the Nakagami distribution given in \eqref{Nakagami_PDF}. Furthermore, the correlation parameter between the channels $\left|g_{0}\right|$ and $\left|g_{k}\right|$ is found as $\rho_{\left|g_{0}\right|,\left|g_{k}\right|} = \mu^2$.
\end{theorem}

\begin{proof}
See Appendix \ref{appendix:Proof_g_k_Nakagami_RV}.
\end{proof}

With the distribution and correlation parameters in Theorem \ref{g_k_Nakagami_RV}, the following theorem presents the joint PDF of the received MRT signal in the FAS across all ports.

\begin{theorem}
The joint PDF of the MRT precoded channels among the $K$ ports $\left|g_{1}\right|, \left|g_{2}\right|, \dots, \left|g_{K}\right|$ is given by
\begin{multline}\label{joint_PDF_Nakagami}
{f_{\left| {{g_1}} \right|, \dots ,\left| {{g_K}} \right| }}\left( {{\tau_1}, \dots ,{\tau_K}} \right)\\
=\int_{0}^{\infty} \frac{2\omega^{\omega}\tau_{0}^{2\omega-1}e^{-\frac{\omega}{\iota}\tau_{0}^2}}{\Gamma(\omega)\iota^{\omega}}\\
\times\prod\limits_{k = 1}^K {\frac{{2\omega \tau_0^{1 - \omega}\tau_k^\omega{e^{-\frac{\mu^2\omega}{{1 - \mu^2}}\frac{{\tau_0^2}}{{\iota}}}}{e^{ - \frac{\omega}{{1 - {\mu}^2}}\frac{{\tau_k^2}}{{\iota}}}}}}{{\iota\left( {1 - \mu^2} \right){ \mu ^{\omega - 1}}}}} \\
\times {I_{\omega - 1}}\left[ {\frac{{2\omega{\mu}{\tau_0}{\tau_k}}}{{\left( {1 - \mu^2} \right)\iota}}} \right] d\tau_{0},
\end{multline}
where $I_{\omega-1}(\cdot)$ is the modified Bessel function of the first kind.
\end{theorem}

\begin{proof}
According to \cite[(B.1) in Proposition 2]{joint_Nakagamai_PDF}, conditioned on $\left| g_{0}\right|$, the joint PDF of $\left|g_{1}\right|,\dots, \left|g_{K}\right|$ is given by
\begin{multline} 
{f_{\left| {{g_1}} \right|, \dots ,\left| {{g_K}} \right|\left| {\left| {{g_0}} \right|} \right.}}\left( {{\tau_1}, \dots ,{\tau_K}|{\tau_0}} \right)\\
=\left[\frac{2\omega\tau_{0}^{1-\omega}e^{-\frac{\omega\mu^2}{\iota(1-\mu^2)}\tau_{0}^2}}{\iota(1-\mu^2)\mu^{\omega-1}}\right]^{K}\\
\times\prod_{k = 1}^K \tau_{k}^{\omega} e^{-\frac{\omega\tau_{k}^2}{\iota(1-\mu^2)}} {I_{\omega - 1}}\left[ {\frac{{2\omega{\mu}{\tau_0}{\tau_k}}}{{\iota\left( {1 - \mu^2} \right)}}} \right].
\end{multline}     
With the PDF of $\left| {{g_0}} \right|$ given in \eqref{Nakagami_PDF}, the unconditional joint PDF can be expressed as \eqref{joint_PDF_Nakagami}, which completes the proof.
\end{proof}

\begin{theorem}\label{joint_CDF_signal_channel}
The CDF of the MRT precoded channels among the $K$ ports $\left|g_{1}\right|, \dots, \left|g_{K}\right|$ can be expressed as
\begin{multline}\label{joint_CDF_signal_channel_formula}
{F_{\left| {{g_1}} \right|, \dots ,\left| {{g_K}} \right| }}\left( {{\tau_1}, \dots ,{\tau_K}} \right)\\
= \frac{2\omega^{\omega}}{\Gamma(\omega)\iota^{\omega}}\int_{0}^{\infty}t^{2\omega-1} e^{-\frac{\omega t^2}{\iota}}\\
\times \prod_{k=1}^{K}\left[1-Q_{\omega}\left(\sqrt{\frac{2\omega\mu^2t^2}{\iota(1-\mu^2)}},\sqrt{\frac{2\omega\tau_{k}^2}{\iota(1-\mu^2)}}\right)\right]dt,
\end{multline}
where $Q_{\omega}(\cdot,\cdot)$ is the generalized Marcum $Q$-function. 
\end{theorem}

\begin{proof}
This comes directly from \cite [Theorem 1]{joint_Nakagamai_PDF}.
\end{proof}

\subsection{The $f$-FAMA Case}
In this subsection, we report the main results when the $f$-FAMA strategy is adopted at each user. Before we present the results, we find the following lemma useful.

\begin{lemma}\label{Lemma_f_Interference}
In $f$-FAMA, the random variable, $|g_{k}^{\rm I}|$, follows a Rayleigh distribution, with the PDF given by
\begin{align}\label{FAMA_interference}
f_{|g_{k}^{\rm I}|}(\tau_{k}) = \frac{2\tau_{k}}{\sigma_{\rm I}^2}e^{-\frac{\tau_{k}^2}{\sigma_{\rm I}^2}}  ,
\end{align}
where $\sigma_{\rm I}^2 = \sum_{i=1}^{U}r_{i}^{-\alpha}\sigma^2\sigma_{\rm s}^2$. Also, the correlation parameter between $|g_{0}^{\rm I}|$ and $|g_{k}^{\rm I}|$ is expressed as $\rho_{|g_{0}^{\rm I}|, |g_{k}^{\rm I}|} = \mu^2$.
\end{lemma}

\begin{theorem}\label{Theorem_Out_f_FAMA}
For the MRT-$f$-FAMA network, the outage probability is given by \eqref{Outage_f_FAMA_Closedform}, as shown at the top of next page, where $a = \omega$ and $(\cdot)_{p}$ represents the Pochhammer symbol. In addition, $\omega$ and $\iota$ have been given in (\ref{eq:omega}).
\begin{figure*}[]
\begin{align}\label{Outage_f_FAMA_Closedform}
{\mathcal O}_{{\rm MRT}\text{-}f\text{-}{\rm FAMA}}
&=\frac{2\omega^{\omega}}{\Gamma(\omega)\iota^{\omega}}\int_{t_0 = 0}^{\infty}\frac{2t_{0}}{\sigma_{\rm I}^2}e^{-\frac{t_{0}^2}{\sigma_{\rm I}^2}}\int_{t = 0}^{\infty}t^{2\omega-1}e^{-\frac{\omega t^2}{\iota}} \nonumber \\ 
&\quad \times\Bigg[          Q_{1}\left(\sqrt{\frac{2\omega\mu^2\Theta_{\rm f}^2t_{0}^2}{(1-\mu^2)(\iota +\omega\Theta_{\rm f}^2\sigma_{\rm I}^2)}},\sqrt{\frac{2\omega\mu^2t^2}{(1-\mu^2)(\iota+\omega\Theta_{\rm f}^2\sigma_{\rm I}^2)}} \right)-\frac{1}{\mu t_{0}}\left(\frac{\mu\iota t_{0}}{\iota+\omega\Theta_{\rm f}^2\sigma_{\rm I}^2} \right)^a \nonumber \\
&\quad \times \left( \frac{\mu t}{\Theta_{\rm f}} \right)^{1-\omega} e^{-\frac{(\omega\mu^2\Theta_{\rm f}^2t_{0}^2+\omega\mu^2t^2)}{(1-\mu^2)(\iota+\omega\Theta_{\rm f}^2\sigma_{\rm I}^2)}} \sum_{q=0}^{a-1}\sum_{p=0}^{a-q-1}  \left(\frac{\sigma_{\rm I}^2 t}{t_{0}} \sqrt{\frac{\omega (1-\mu^2)}{2\iota }} \right)^{q+p} \left(\frac{\iota+\omega\Theta_{\rm f}^2\sigma_{\rm I}^2}{\Theta_{\rm f}\sigma_{\rm I}^2}\sqrt{\frac{2}{\iota\omega (1-\mu^2)}} \right)^q \nonumber \\
&\quad\times  \left(\sqrt{\frac{2\omega\Theta_{\rm f}^2}{\iota(1-\mu^2)}} \right)^{p}\frac{(a-q-p)_{p}}{p!}I_{1-\omega+q+p}\left( \frac{2\omega\Theta_{\rm f}\mu^2tt_{0}}{(1-\mu^2)(\iota+\omega\Theta_{\rm f}^2\sigma_{\rm I}^2)}\right) 
\Bigg]^{K} dt dt_{0}
\end{align}
\hrulefill
\end{figure*}   
\end{theorem}

\begin{proof}
See Appendix \ref{Proof_Out_f_FAMA}. 
\end{proof}

Note that as $\tau \to \infty$, $I_{1-\omega+q+p}(\tau) \to \infty$, which leads to a numerical computation problem in \eqref{Outage_f_FAMA_Closedform}.  To address this issue, we have the following corollary.

\begin{corollary}\label{Corollary_Out_f_FAMA}
The outage probability in \eqref{Outage_f_FAMA_Closedform} can be numerically evaluated using \eqref{Outage_f_FAMA_No_Besel}, see top of next page.
\begin{figure*}[]
\centering
\begin{align}\label{Outage_f_FAMA_No_Besel}
{\mathcal O}_{{\rm MRT}\text{-}f\text{-}{\rm FAMA}}
&=\frac{2\omega^{\omega}}{\Gamma(\omega)\iota^{\omega}}\int_{t_0 = 0}^{\infty}\frac{2t_{0}}{\sigma_{\rm I}^2}e^{-\frac{t_{0}^2}{\sigma_{\rm I}^2}}\int_{t = 0}^{\infty}t^{2\omega-1}e^{-\frac{\omega t^2}{\iota}} \nonumber \\ 
&\quad \times\Bigg[          Q_{1}\left(\sqrt{\frac{2\omega\mu^2\Theta_{\rm f}^2t_{0}^2}{(1-\mu^2)(\iota +\omega\Theta_{\rm f}^2\sigma_{\rm I}^2)}},\sqrt{\frac{2\omega\mu^2t^2}{(1-\mu^2)(\iota+\omega\Theta_{\rm f}^2\sigma_{\rm I}^2)}} \right)-\frac{1}{\pi\mu t_{0}}e^{-\frac{(\sqrt{\omega}\mu\Theta_{\rm f}t_{0}-\sqrt{\omega}\mu t)^2}{(1-\mu^2)(\iota+\omega\Theta_{\rm f}^2\sigma_{\rm I}^2)}} \nonumber \\
&\quad \times \left(\frac{\mu\iota t_{0}}{\iota+\omega\Theta_{\rm f}^2\sigma_{\rm I}^2} \right)^a\left( \frac{\mu t}{\Theta_{\rm f}} \right)^{1-\omega}
\sum_{q=0}^{a-1}\sum_{p=0}^{a-q-1}  \left(\frac{\sigma_{\rm I}^2 t}{t_{0}} \sqrt{\frac{\omega (1-\mu^2)}{2\iota }} \right)^{q+p} \left(\frac{\iota+\omega\Theta_{\rm f}^2\sigma_{\rm I}^2}{\Theta_{\rm f}\sigma_{\rm I}^2}\sqrt{\frac{2}{\iota\omega (1-\mu^2)}} \right)^q \nonumber \\
&\quad\times  \left(\sqrt{\frac{2\omega\Theta_{\rm f}^2}{\iota(1-\mu^2)}} \right)^{p}\frac{(a-q-p)_{p}}{p!}\int_{0}^{\pi} \cos{(\theta(1-\omega+q+p))}e^{- \frac{4\omega\Theta_{\rm f}\mu^2tt_{0}}{(1-\mu^2)(\iota+\omega\Theta_{\rm f}^2\sigma_{\rm I}^2)}\sin^2{\frac{\theta}{2}}}d\theta 
\Bigg]^{K} dt dt_{0}
\end{align}
\hrulefill
\end{figure*}
\end{corollary}

\begin{proof}
See Appendix \ref{Proof_Corollary_Out_f_FAMA}.
\end{proof}

In the special case $K=1$, the correlation parameter is reduced to $\mu = 1$. However, substituting $\mu = 1$ directly into \eqref{Outage_f_FAMA_No_Besel} leads the denominator being $0$. The following remark provides an expression to address this special case. 

\begin{remark}\label{OP_fixed_antenna_Fast}
Under MRT-$f$-FAMA and $K=1$, the outage probability can be evaluated by 
\begin{align}\label{fixed_user_FAST}
{\mathcal O}_{{\rm MRT}\text{-}f\text{-}{\rm FAMA}}=\int_{0}^{\infty}
\frac{2t}{\sigma_{\rm I}^2}e^{-\frac{t^2}{\sigma_{\rm I}^2}}P\left(\omega,\frac{\omega\Theta_{\rm f}^2}{\iota}t^2\right)dt,
\end{align}
where $P(\cdot,\cdot)$ denotes the regularized lower incomplete gamma function. The expression is obtained by the definition of outage probability with the result in Theorem \ref{g_k_Nakagami_RV} and Lemma \ref{Lemma_f_Interference}. This case corresponds to a user with a single fixed antenna.
\end{remark}

\subsection{The $s$-FAMA Case}
Here, we turn our attention to the $s$-FAMA case and present the distributions that are necessary to obtain the expression for the outage probability given in the following theorem. Since the exact distribution of $\left|g_{k}^{{\rm [s]}} \right|$ is not tractable, we provide an approximate distribution.

\begin{theorem}\label{g_k_s_Nakagami_RV}
In $s$-FAMA, the random variable $\left|g_{k}^{\rm {[s]}} \right|$ follows a Nakagami distribution, with the PDF given by
\begin{align}\label{slow_interference_distribution}
f_{\left|g_{k}^{\rm {[s]}} \right|}(\tau_{k}) =  {\rm Nakagami}(\tau_{k}; \Omega, \phi),
\end{align}
where 
\begin{equation}\label{eq:Omega}
\left\{\begin{aligned}
\Omega &= \frac{\left(\sum_{i=1}^{U}r_{i}^{-\alpha}\right)^2}{\sum_{i=1}^{U}r_{i}^{-2\alpha}},\\
\phi& =  \sigma^2\sum_{i=1}^{U}r_{i}^{-\alpha}. 
\end{aligned}\right.
\end{equation}
Similar to the $f$-FAMA case before, the correlation parameter between $\left|g_{0}^{\rm {[s]}} \right|$ and $\left|g_{k}^{\rm {[s]}} \right|$ is given as $\rho_{\left|g_{0}^{\rm {[s]}} \right|, \left|g_{k}^{\rm {[s]}} \right|} = \mu^2$.
\end{theorem}

\begin{proof}
See Appendix \ref{appendix:Proof_g_k_s_Nakagami_RV}.
\end{proof}

\begin{remark}\label{Remark_Nagami_Rayleigh}
When the shape parameter becomes $\Omega = 1$, Nakagami distribution is reduced to Rayleigh distribution. At the same time, if $U = 1$, then the distribution of the interference channel for MRT-$s$-FAMA, $\left|g_{k}^{\rm [s]}\right|$, becomes identical to that of the interference channel for MRT-$f$-FAMA, $\left|g_{k}^{\rm [f]}\right|$.
\end{remark}

\begin{theorem}\label{theorem_joint_PDF_Slow_Interference}
In $s$-FAMA, the joint PDF of the interference channel among the $K$ ports $\left|g_{1}^{\rm {[s]}} \right|,\dots, \left|g_{K}^{\rm {[s]}} \right|$ is given by
\begin{multline}\label{joint_PDF_Nakagami_interference}
{f_{\left|g_{1}^{\rm {[s]}} \right|, \dots ,\left|g_{K}^{\rm {[s]}} \right| }}\left( {{\tau_1}, \dots ,{\tau_K}} \right)\\
=\int_{0}^{\infty} \frac{2\Omega^{\Omega}\tau_{0}^{2\Omega-1}e^{-\frac{\Omega}{\phi}\tau_{0}^2}}{\Gamma(\Omega)\phi^{\Omega}}\\
\times\prod\limits_{k = 1}^K {\frac{{2\Omega \tau_0^{1 - \Omega}\tau_k^\Omega{e^{-\frac{\mu^2\Omega}{{1 - \mu^2}}\frac{{\tau_0^2}}{{\phi}}}}{e^{ - \frac{\Omega}{{1 - {\mu}^2}}\frac{{\tau_k^2}}{{\phi}}}}}}{{\phi\left( {1 - \mu^2} \right){ \mu ^{\Omega - 1}}}}} \\
\times {I_{\Omega - 1}}\left[ {\frac{{2\Omega{\mu}{\tau_0}{\tau_k}}}{{\left( {1 - \mu^2} \right)\phi}}} \right] d\tau_{0},
\end{multline}
where $\Omega$ and $\phi$ are given in (\ref{eq:Omega}).
\end{theorem}

\begin{proof}
From Theorem \ref{g_k_s_Nakagami_RV}, the interference term $\left|g_{k}^{\rm {[s]}} \right|$ is formulated with the Nakagami distribution, same as $\left|g_{k}\right|$, while the correlation parameter remains unchanged. Consequently, by substituting $\omega = \Omega$ and $\iota = \phi$ into \eqref{joint_PDF_Nakagami}, the joint PDF of  $\left|g_{1}^{\rm {[s]}} \right|,\dots, \left|g_{K}^{\rm {[s]}} \right|$ can be expressed as \eqref{joint_PDF_Nakagami_interference}.
\end{proof}

\begin{theorem}\label{Theorem_Out_s_FAMA}
Under the MRT-$s$-FAMA scenario, while $\Omega > \omega$, the outage probability is given by \eqref{Outage_S_FAMA_Closedform}, as shown on next  page, where $b = \omega+\Omega-1$.
\begin{figure*}[]
\begin{align}\label{Outage_S_FAMA_Closedform}
{\mathcal O}_{{\rm MRT}\text{-}s\text{-}{\rm FAMA}}
= &\frac{4\omega^{\omega}\Omega^{\Omega}}{\Gamma(\omega)\Gamma(\Omega)\iota^{\omega}\phi^{\Omega}}\int_{t_0 = 0}^{\infty}t_{0}^{2\Omega-1}e^{-\frac{\Omega t_{0}^2}{\phi}}\int_{t = 0}^{\infty}t^{2\omega-1}e^{-\frac{\omega t^2}{\iota}} \nonumber \\
&\times\Bigg[Q_{\Omega}\left( \sqrt{\frac{2\omega\Omega\mu^2\Theta_{\rm s}^2 t_{0}^2}{(1-\mu^2)(\iota\Omega+\omega\phi\Theta_{\rm s}^2)}}, \sqrt{\frac{2\omega\Omega\mu^2t^2}{(1-\mu^2)(\iota\Omega+\omega\phi\Theta_{\rm s}^2)}}\right)-\left(\frac{1}{\mu t_{0}}\right)^{\Omega} \left( \frac{\mu t}{\Theta_{\rm s}} \right)^{1-\omega}\nonumber \\
& \times \left(\frac{\mu\iota\Omega t_{0} }{\iota\Omega+\omega\phi\Theta_{\rm s}^2}\right)^{b}  e^{-\frac{\omega\Omega\mu^2\Theta_{\rm s}^2t_{0}^2+\omega\Omega\mu^2  t^2}{(1-\mu^2)(\iota\Omega+\omega\phi\Theta_{\rm s}^2)} }\sum_{q=0}^{b-1}\sum_{p=0}^{b-q-1}  \left(\frac{\phi t}{\Omega t_{0}} \sqrt{\frac{\omega (1-\mu^2)}{2\iota }} \right)^{q+p}\left(\frac{\iota\Omega+\omega\phi\Theta_{\rm s}^2}{\phi\Theta_{\rm s}}\sqrt{\frac{2}{ \iota\omega(1-\mu^2)}} \right)^q  \nonumber \\
& \times  \left(\sqrt{\frac{2\omega\Theta_{\rm s}^2}{\iota(1-\mu^2)}} \right)^{p}\frac{(b-q-p)_{p}}{p!}I_{1-\omega+q+p}\left( \frac{2\omega\Omega\Theta_{\rm s}\mu^2 t t_{0}}{(1-\mu^2)(\iota\Omega+\omega\phi\Theta_{\rm s}^2)}\right) \Bigg]^K dt dt_{0}
\end{align}
\hrulefill
\end{figure*}
\end{theorem}

\begin{proof}
See Appendix \ref{Proof_Out_s_FAMA}. 
\end{proof}

Similar to the discussion in MRT-$f$-FAMA earlier, as $\tau \to \infty$, $I_{1-\omega+q+p}(\tau) \to \infty$, \eqref{Outage_S_FAMA_Closedform} cannot be numerically calculated. The following corollary can be used to evaluate \eqref{Outage_S_FAMA_Closedform}. 

\begin{corollary}\label{Corollary_Out_s_FAMA}
The outage probability in \eqref{Outage_S_FAMA_Closedform} can be expressed by \eqref{Outage_S_FAMA_No_Besel}, shown at the bottom of next page.
\begin{figure*}[]
\begin{align}\label{Outage_S_FAMA_No_Besel}
{\mathcal O}_{{\rm MRT}\text{-}s\text{-}{\rm FAMA}}
= & \frac{4\omega^{\omega}\Omega^{\Omega}}{\Gamma(\omega)\Gamma(\Omega)\iota^{\omega}\phi^{\Omega}}\int_{t_0 = 0}^{\infty}t_{0}^{2\Omega-1}e^{-\frac{\Omega t_{0}^2}{\phi}} \int_{t = 0}^{\infty}t^{2\omega-1}e^{-\frac{\omega t^2}{\iota}} \nonumber \\
\times&\Bigg[Q_{\Omega}\left( \sqrt{\frac{2\omega\Omega\mu^2\Theta_{\rm s}^2 t_{0}^2}{(1-\mu^2)(\iota\Omega+\omega\phi\Theta_{\rm s}^2)}}, \sqrt{\frac{2\omega\Omega\mu^2t^2}{(1-\mu^2)(\iota\Omega+\omega\phi\Theta_{\rm s}^2)}}\right)-\frac{1}{\pi}\left(\frac{1}{\mu t_{0}}\right)^{\Omega}
\left( \frac{\mu t}{\Theta_{\rm s}} \right)^{1-\omega}\nonumber \\
\times& \left(\frac{\mu\iota\Omega t_{0} }{\iota\Omega+\omega\phi\Theta_{\rm s}^2}\right)^{b}e^{-\frac{(\sqrt{\omega\Omega}\mu \Theta_{\rm s} t_{0}-\sqrt{\omega\Omega}\mu t)^2}{(1-\mu^2)(\iota\Omega+\omega\phi\Theta_{\rm s}^2)} }\sum_{q=0}^{b-1}\sum_{p=0}^{b-q-1}  \left(\frac{\phi t}{\Omega t_{0}} \sqrt{\frac{\omega (1-\mu^2)}{2\iota }} \right)^{q+p}  \left(\frac{\iota\Omega+\omega\phi\Theta_{\rm s}^2}{\phi\Theta_{\rm s}}\sqrt{\frac{2}{ \iota\omega(1-\mu^2)}} \right)^q\nonumber \\
\times&  \left(\sqrt{\frac{2\omega\Theta_{\rm s}^2}{\iota(1-\mu^2)}} \right)^{p}\frac{(b-q-p)_{p}}{p!}\int_{0}^{\pi} \cos{(\theta(1-\omega+q+p))}e^{- \frac{4\omega\Omega\Theta_{\rm s}\mu^2tt_{0}}{(1-\mu^2)(\iota\Omega+\omega\phi\Theta_{\rm s}^2)}\sin^2{\frac{\theta}{2}}}d\theta   \Bigg]^K dt dt_{0}
\end{align}
\hrulefill
\end{figure*}
\end{corollary}

\begin{proof}
See Appendix \ref{Proof_Corollary_Out_s_FAMA}.
\end{proof}

\begin{remark}\label{OP_fixed_antenna_Slow}
Under the MRT-$s$-FAMA scenario, the outage probability with $K=1$, or with users equipped with a single fixed antenna, is given by 
\begin{align}\label{fixed_user_Slow}
{\mathcal O}_{{\rm MRT}\text{-}s\text{-}{\rm FAMA}}
= \int_{0}^{\infty}\frac{2{\Omega}^{\Omega}t^{2\Omega-1}e^{-\frac{\Omega t^2}{\phi}}}{\Gamma(\Omega)\phi^{\Omega}}
P\left(\omega,\frac{\omega\Theta_{\rm s}^2}{\iota}t^2\right)dt.
\end{align}
This expression can be derived by the definition of outage probability with the result in Theorems \ref{g_k_Nakagami_RV} and \ref{g_k_s_Nakagami_RV}.
\end{remark}

\subsection{Special Cases} 
The expressions \eqref{Outage_f_FAMA_No_Besel} and \eqref{Outage_S_FAMA_No_Besel} provide the outage probability under interference-limited conditions, where noise power can be neglected. The following remark addresses the contrasting scenario, in which noise power dominates and the interference power can be safely ignored.

\begin{remark}

Considering the use of a non-dense frequency reuse scheme so that the channels for the users are orthogonal to each other, thereby resulting in a noise-limited system, the outage probability for both fast port switching and slow port switching can be expressed as
\begin{multline}\label{Outage_SNR}
{\mathcal O}_{\rm MRT-SNR} 
=\frac{2\omega^{\omega}}{\Gamma(\omega)\iota^{\omega}} \int_{0}^{\infty} t^{2\omega-1} e^{-\frac{\omega t^2}{\iota}} \\
\times \left[ 1 - Q_{\omega} \left( \sqrt{\frac{2\omega \mu^2 t^2}{\iota (1-\mu^2)}}, \sqrt{\frac{2\omega \sigma^2_\eta \gamma_{\mathrm{th}}^{\rm SNR}}{\iota \sigma_s^2 (1-\mu^2)}} \right) \right]^K dt,
\end{multline}
in which $\gamma_{\mathrm{th}}^{\rm SNR}$ is the preset threshold for the noise-limited scenario. This expression is the result after substituting $\tau_{1} = \cdots = \tau_{K} = \sqrt{\frac{\sigma^{2}_\eta \gamma_{\mathrm{th}}^{\rm SNR}}{ \sigma_{s}^2}}$ into \eqref{joint_CDF_signal_channel_formula}.
\end{remark}

\begin{remark}Under the noise-limited scenario, given $K=1$, or the user equipped with a single fixed antenna, the outage probability is given by
\begin{align}\label{fixed_user_SNR}
{\mathcal O}_{\rm MRT-SNR} = P\left(\omega,\frac{\omega\sigma_{\eta}^2\gamma_{\rm th}^{\rm SNR}}{\iota\sigma_{s}^2}\right).
\end{align}
This expression can be derived from the definition of outage probability and the result in Theorem \ref{g_k_Nakagami_RV}.   
\end{remark}

The expressions \eqref{Outage_f_FAMA_Closedform} and \eqref{Outage_f_FAMA_No_Besel} generalize the results in \cite{fast_FAMA} by extending the scenario to multiple transmitter antennas and varying large-scale fading, while \eqref{Outage_S_FAMA_Closedform} and \eqref{Outage_S_FAMA_No_Besel} generalize the results in \cite{slow_FAMA} with the same extensions. The same discussion applies to \eqref{Outage_SNR} in relation to the results in \cite{Fluid_antenna_system}. Notwithstanding, the expressions \eqref{Outage_f_FAMA_Closedform} and \eqref{Outage_f_FAMA_No_Besel} are slightly different from \cite[(17) \& (19)]{fast_FAMA} as different channel models and derivation methods were considered, which is the same for expressions \eqref{Outage_SNR} and \cite[(16)]{Fluid_antenna_system}. However, they still achieve the same result when considering a single transmitter antenna, ignoring large-scale fading, and employing the same correlation parameter $\mu$. Furthermore, the expressions \eqref{Outage_S_FAMA_Closedform} and \eqref{Outage_S_FAMA_No_Besel} can be simplified to \cite[(21) \& (22)]{slow_FAMA} under the assumptions of a single transmitter antenna and neglecting the large-scale fading effects.

\section{Numerical Results and Discussion}\label{sec:results}
Here, we provide the simulation results to evaluate the network performance of MRT-FAMA systems with fast and slow port switching. In the simulations, we set $\sigma  = \sigma_{s} = 1$, $\sigma_{\eta}= 10^{-4}$ and $W = 5$. The path loss exponent, $\alpha = 3$, is specified for an urban area cellular radio environment \cite[Table 3.2]{Path_Loss_Exponent}. The network parameters are also set as $N = 2$, $U = 3 $, $\gamma^{\rm [f]}_{\rm th} = \gamma^{\rm [s]}_{\rm th} = 1$, and $r_{0} = r_{i} = 100, \forall i$ unless otherwise specified. We provide Monte-Carlo simulation results and those obtained using the expressions \eqref{Outage_f_FAMA_No_Besel} and \eqref{Outage_S_FAMA_No_Besel}. Also, the results for using only a fixed antenna at each user (\ref{fixed_user_Slow}) are included as benchmark for comparison. 
			
Fig.~\ref{fig:Monte_SIR} provides the numerical results for the outage probability against the SIR threshold when there are $U=3$ interfering BSs and each user has a FAS with $K=10$ ports. First of all, the results demonstrate that the analytical expressions \eqref{Outage_f_FAMA_No_Besel} and \eqref{Outage_S_FAMA_No_Besel} align closely with the Monte-Carlo results, validating our analysis for both MRT-$f$-FAMA and MRT-$s$-FAMA. The slight discrepancy seen in the case of MRT-$s$-FAMA is due to the Gamma approximation in Theorem \ref{g_k_s_Nakagami_RV}. Also, note that the Gamma approximation in Theorem \ref{g_k_s_Nakagami_RV} will become exact if the interference distances $r_{i}$ are identical. Moreover, given two transmit antennas on each BS ($N = 2$), under the same threshold, MRT-$f$-FAMA outperforms MRT-$s$-FAMA, which is expected by using a much faster port switching strategy.
			
\begin{figure}
\centering
\includegraphics[width=1.0\linewidth]{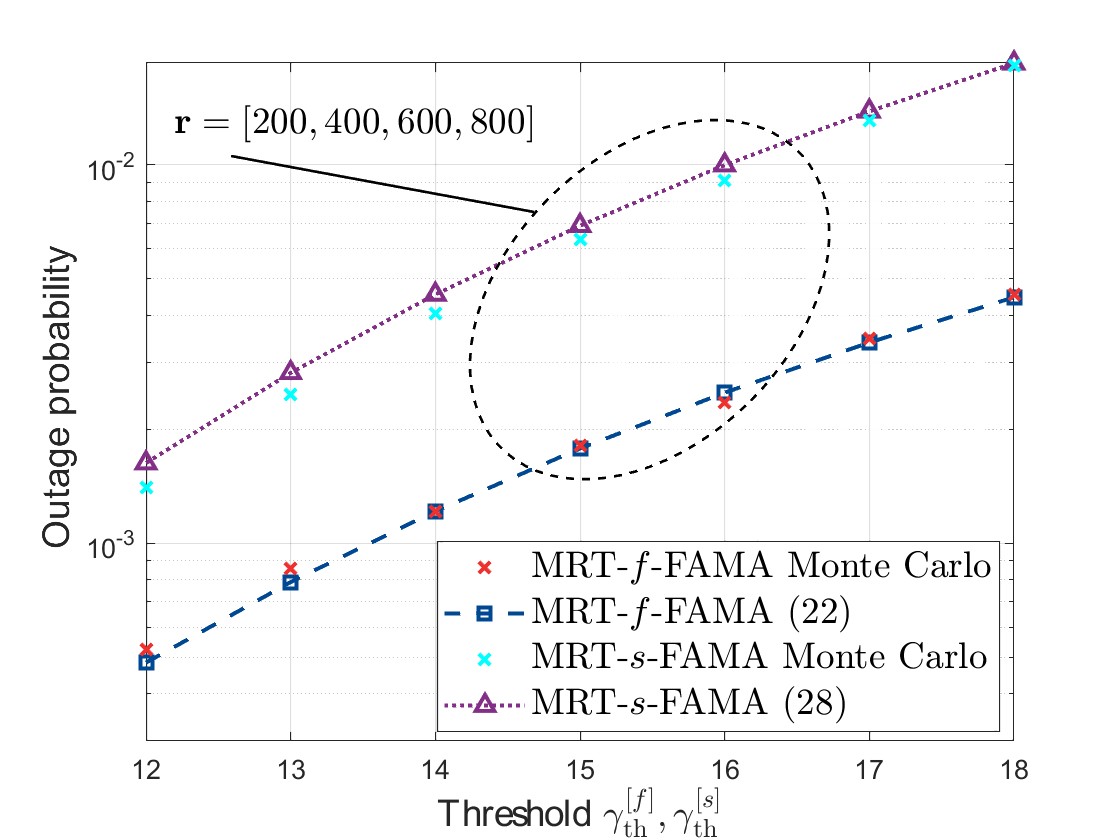}
\caption{Outage probability against the SIR threshold $\gamma^{\rm [f]}_{\rm th}, \gamma^{\rm [s]}_{\rm th}$, when $\bold{r} = [200,400,600,800]$, and the number of ports $K = 10$.}\label{fig:Monte_SIR}
\end{figure} 
			
As the results in Fig.~\ref{fig:Monte_SIR} have confirmed the accuracy of the analytical results, we now turn to the results of Fig.~\ref{fig:f_s_FAMA_N} to study the performance between MRT-$f$-FAMA and MRT-$s$-FAMA when the number of transmit antennas per BS, $N$, varies. The subplot displays the PDF of the signal \textcolor{black}{magnitude} and that for the interference \textcolor{black}{magnitude} with $f$-FAMA and $s$-FAMA when $N = 8$. As before, it is seen that MRT-$f$-FAMA outperforms MRT-$s$-FAMA if $N<3$. However, as $N$ increases, crossovers occur meaning that MRT-$s$-FAMA becomes a better option if $N$ is large. Specifically, when the SIR threshold $\gamma^{\rm [f]}_{\rm th} = \gamma^{\rm [s]}_{\rm th} = 18$, the crossover happens at $N=5$, but if $\gamma^{\rm [f]}_{\rm th} = \gamma^{\rm [s]}_{\rm th} = 14$, this happens at $N=4$. This indicates that as the SIR threshold increases, the number of transmit antennas per BS $N$ required for MRT-$s$-FAMA to outperform MRT-$f$-FAMA increases. 
			
This observation is counter-intuitive and deviates from our expectation that $f$-FAMA should always perform better than $s$-FAMA \cite{slow_FAMA}. Notably, this phenomenon only occurs when $N$ is large and in this case, the outage probability is relatively low, below $10^{-1}$. This can be elucidated by the PDFs of the signal and interference \textcolor{black}{magnitudes}, as depicted in the subplot. As we can see, given $N = 8$, the mean signal \textcolor{black}{magnitude} is much greater than that of the interference \textcolor{black}{magnitude} for both $f$-FAMA and $s$-FAMA cases. Additionally, the interference \textcolor{black}{magnitude} for $f$-FAMA has a lower mean value but higher variance, whereas for the $s$-FAMA scenario, the interference \textcolor{black}{magnitude} exhibits an opposite behavior. To understand that, based on Lemma \ref{Lemma_f_Interference}, the mean and variance of the interference \textcolor{black}{magnitude} for $f$-FAMA can be found as
\begin{equation}
\left\{\begin{aligned}
\mathrm{E}\left[\left|g_{k}^{\rm I}\right|\right]& = \frac{\sigma_{\rm I}}{2}\sqrt{\pi},\\
\mathrm{Var}\left[\left|g_{k}^{\rm I}\right|\right] &= \sigma_{\rm I}^2 \left(\frac{4-\pi}{4}\right).
\end{aligned} \right.
\end{equation}
By contrast, according to Theorem \ref{g_k_s_Nakagami_RV} and \cite[(17) \& (18)]{nakagami1960m}, the mean and variance of the interference \textcolor{black}{magnitude} for $s$-FAMA can be expressed as
\begin{equation}
\left\{\begin{aligned}
\mathrm{E}\left[\sigma_{s}\left|g_{k}^{\rm {[s]}} \right|\right] &= \sigma_{s}\frac{\Gamma(\Omega+\frac{1}{2})}{\Gamma(\Omega)}\sqrt{\frac{\phi}{\Omega}},\\
\mathrm{Var}\left[\sigma_{s}\left|g_{k}^{\rm {[s]}} \right|\right] &= \sigma_{s}^2\phi\left(1-\frac{1}{\Omega}\left(\frac{\Gamma(\Omega+\frac{1}{2})}{\Gamma(\Omega)}\right)^2 \right) \approx \sigma_{s}^2\frac{\phi}{5\Omega}.
\end{aligned} \right.
\end{equation}
Consequently, the ratio between the mean interference \textcolor{black}{magnitude} of $f$-FAMA and that of $s$-FAMA is given by
\begin{align}
\frac{\mathrm{E}\left[\left|g_{k}^{\rm I}\right|\right]}{\mathrm{E}\left[\sigma_{s}\left|g_{k}^{\rm {[s]}} \right|\right]} = \frac{\sqrt{\pi}}{2}\frac{\Gamma(\Omega)\sqrt{\Omega}}{\Gamma(\Omega+\frac{1}{2})}.
\end{align}
Notice that $\frac{\Gamma(\Omega)\sqrt{\Omega}}{\Gamma(\Omega+\frac{1}{2})}$ is monotonically decreasing with respect to $\Omega$. Considering a single interfering BS case ($\Omega = 1$), the ratio $ \frac{\mathrm{E}\left[\left|g_{k}^{\rm I}\right|\right]}{\mathrm{E}\left[\left|g_{k}^{\rm {[s]}} \right|\right]}$ reaches its maximum value of 1.  Therefore, it can be proved that the mean interference \textcolor{black}{magnitude} for $f$-FAMA is lower than that for $s$-FAMA. Similarly, the ratio between the variances is given by
\begin{align}
\frac{\mathrm{Var}\left[\left|g_{k}^{\rm I}\right|\right]}{\mathrm{Var}\left[\sigma_{s}\left|g_{k}^{\rm {[s]}} \right|\right]} = 5\Omega\left(\frac{4-\pi}{4}\right).
\end{align}
Given $\Omega > 1$, it follows that the variance of the interference \textcolor{black}{magnitude} in $f$-FAMA is higher than that in $s$-FAMA.
			
To sum up, while the mean signal \textcolor{black}{magnitude} is much larger than the mean interference \textcolor{black}{magnitude}, the higher variance in $f$-FAMA could degrade its performance more than in $s$-FAMA when the number of transmit antennas per BS is large. 
				
\begin{figure}
\centering
\includegraphics[width=1.0\linewidth]{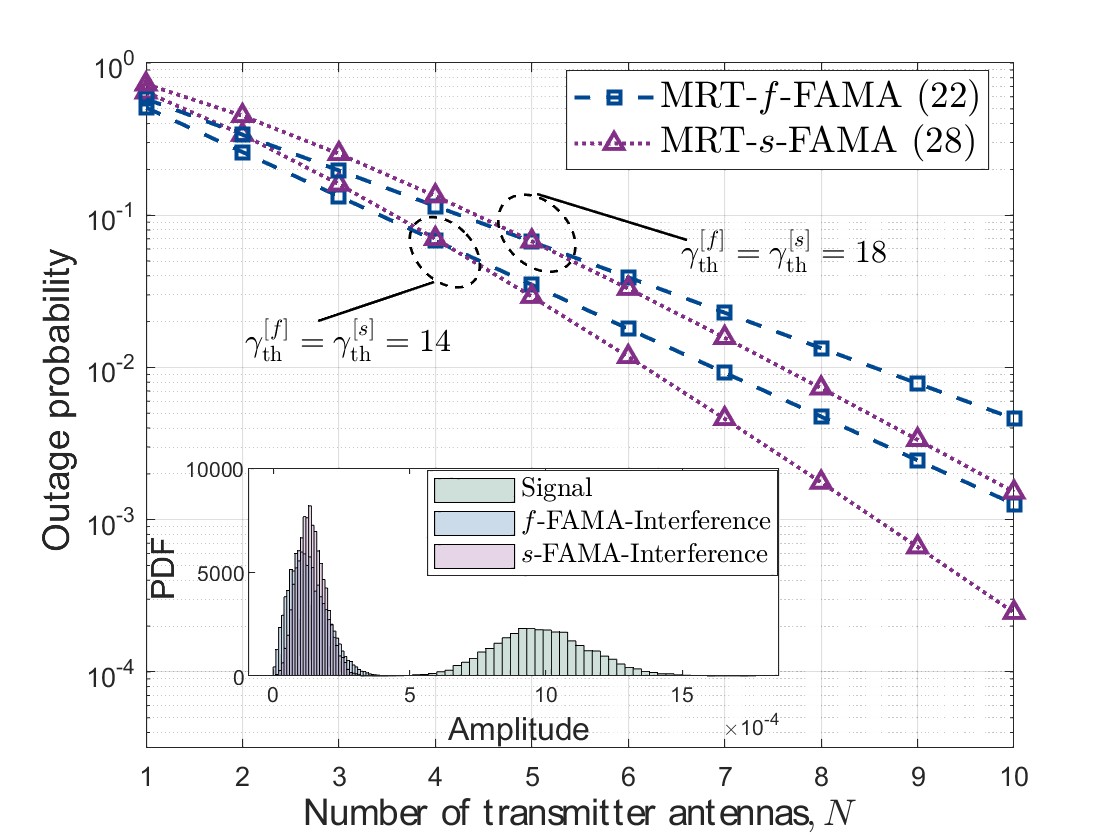}
\caption{Outage probability against the number of transmit antennas per BS $N$, when $\bold{r} = [200,400,600,800]$, the number of ports $K= 2$. The subplot shows the PDF of the signal \textcolor{black}{magnitude} and various interference \textcolor{black}{magnitude} when $K=1$ and $N=8$.}\label{fig:f_s_FAMA_N}
\end{figure}
			
Now, the results in Fig.~\ref{fig:f_s_FAMA_K} investigates the outage probability performance for both MRT-$f$-FAMA and MRT-$s$-FAMA when the number of FAS ports, $K$, at each user changes. We also provide the results for the case when each user uses one fixed antenna, instead of FAS, for comparison. The results illustrate that to achieve an outage probability below $10^{-3}$, a user with a single fixed antenna needs $15$ transmit antennas per BS for MRT, whereas a user with a $12$-port FAS needs only $2$ transmit antennas per BS for MRT. In the case with with $4$ transmit BS antennas, a user equipped with a $20$-port FAS can achieve an outage probability below $10^{-8}$, but a fixed antenna user requires $35$ BS antennas to attain the same performance. This shows a huge advantage of using FAS at each user in reducing the burden of BS. On the other hand, the results also indicate that as the number of ports $K$ increases, the performance gap between MRT-$f$-FAMA and MRT-$s$-FAMA widens but with the increase of the number of BS antennas $N$, the performance gap decreases, as has been reported earlier.

\begin{figure}
\centering
\includegraphics[width=1.0\linewidth]{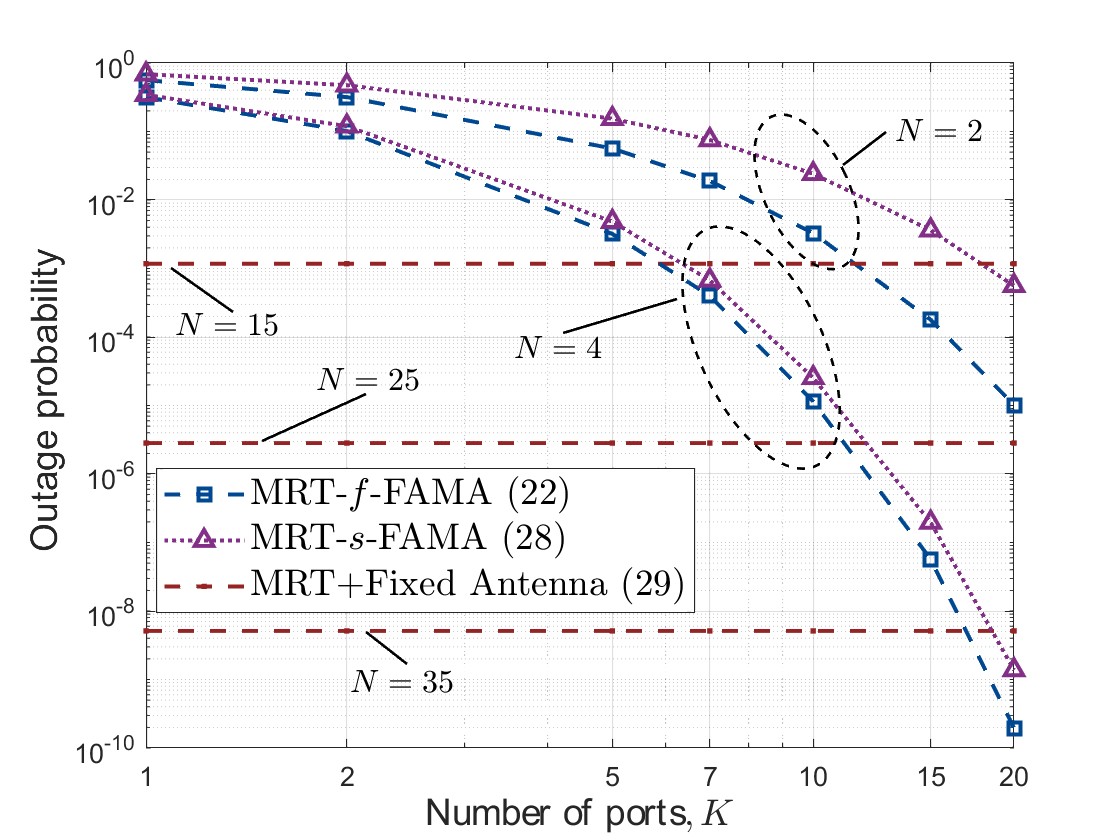}
\caption{Outage probability for different approaches against the number of ports $K$, given different number of transmit antennas per BS $N$.}\label{fig:f_s_FAMA_K}
\end{figure}

\begin{figure}
\centering
\includegraphics[width=1.0\linewidth]{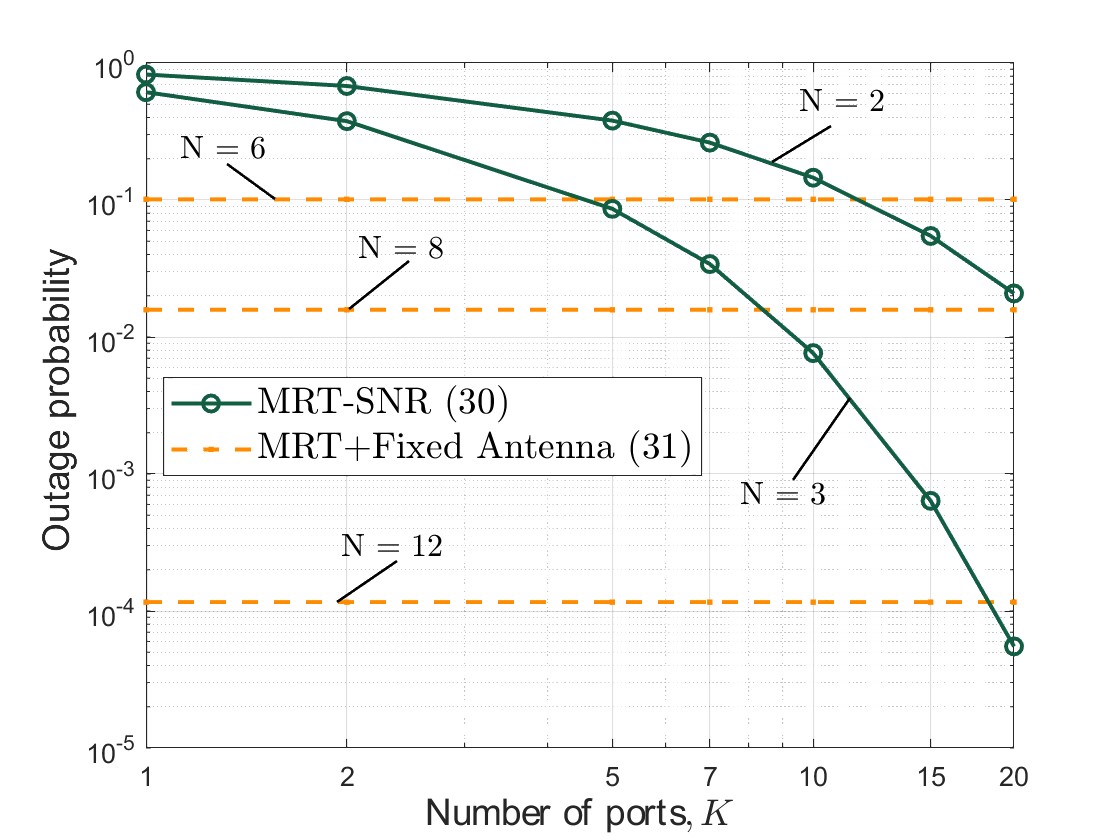}
\caption{Outage probability against the number of ports $K$, given different number of transmit antennas per BS $N$, under a noise-limited system with a signal-to-noise ratio threshold of $\gamma_{\mathrm{th}}^{\rm SNR} = 25$ dB.}\label{fig:OP_K_SNR}
\end{figure}

Similar results are presented in Fig.~\ref{fig:OP_K_SNR}, where the outage probability is plotted against the number of FAS ports $K$ at each user. For comparison, the performance of a user with a fixed antenna is also included. The results demonstrate that for a user with 20-port FAS, only $2$ transmit antennas per BS for MRT are enough to get an outage probability of approximately $10^{-2}$. The fixed-antenna user by contrast requires $4$ transmit antennas per BS to achieve the same performance. To achieve an outage probability below $10^{-4}$, a user with 20-port FAS requires $3$ transmit antennas per BS, whereas the fixed-antenna user needs $6$ antennas per BS. It is noticed that the reduction in the number of transmitter antennas per BS is not as significant compared to previous results where interference exists. Despite this, the presence of FAS at the user side still reduces the burden at the BS under noise-limited conditions. 

The results in Fig.~\ref{fig:f_s_FAMA_W} examine the effects of different FAS sizes on the outage probability when the number of interfering BSs, $U$, changes. Notably, when there is only one interfering BS ($U = 1$), the performance of MRT-$f$-FAMA and MRT-$s$-FAMA are identical. This occurs because with $U = 1$, the interference model in MRT-$s$-FAMA simplifies to that in MRT-$f$-FAMA, as explained in Remark \ref{Remark_Nagami_Rayleigh}. Furthermore, with $K = 60$ and two interfering BSs ($U = 2$), the outage probabilities for both MRT-$f$-FAMA and MRT-$s$-FAMA converge at $W = 3 \lambda$. In contrast, with $U = 1$, the outage probability does not converge until $W = 5 \lambda$. This observation suggests that the performance of FAS is influenced not only by the number of ports $K$ but also significantly by its size $W$. Moreover, while the number of interfering BSs is relatively small, increasing $W$ for a given $K$ can also enhance the performance.

\begin{figure}
\centering
\includegraphics[width=1.0\linewidth]{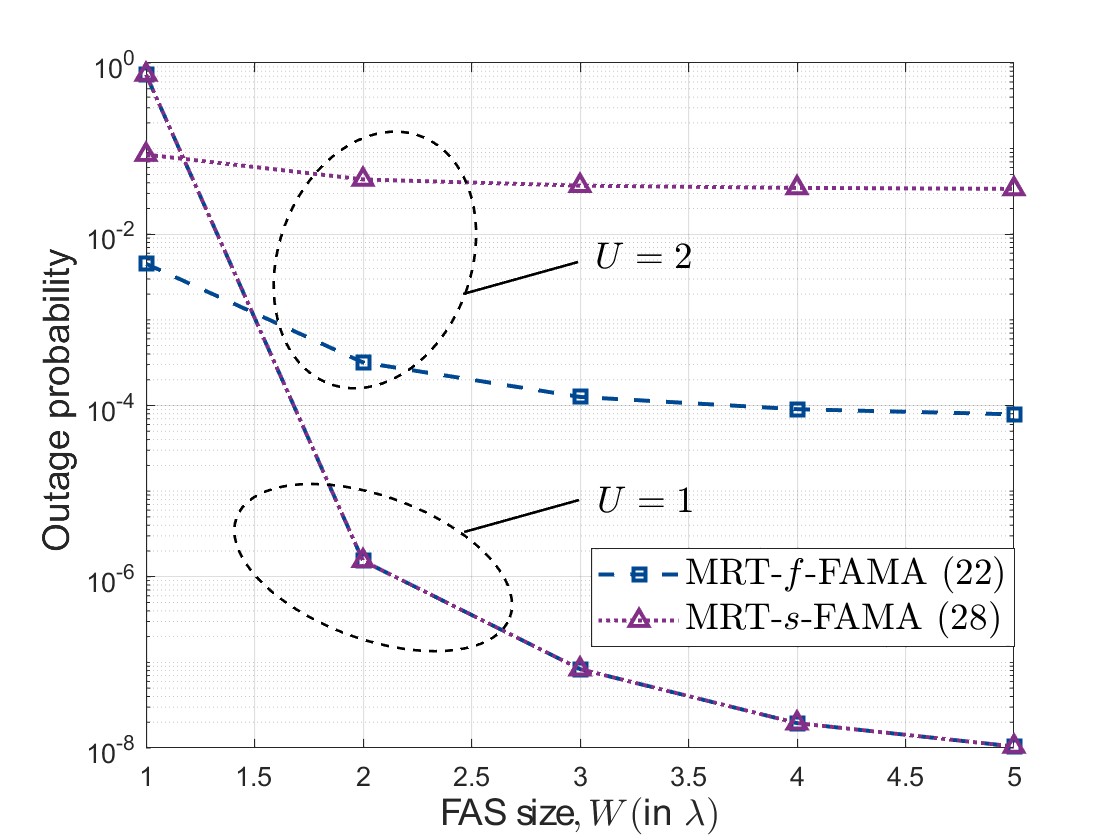}
\caption{Outage probability against the normalized FAS size $W$ (in the number of $\lambda$), given different number of interfering BSs $U$, when the number of ports $K= 60$ and the SIR threshold $\gamma^{\rm [f]}_{\rm th}= \gamma^{\rm [s]}_{\rm th} = 6$.}\label{fig:f_s_FAMA_W}
\end{figure} 

Finally, Fig.~\ref{fig:OP_U_threshold_3} shows how the outage probability is affected by the number of interfering BSs $U$, when the number of ports $K$ is set to $25$ and $30$, respectively. Apparently, as $U$ increases, the performance for both MRT-$f$-FAMA and MRT-$s$-FAMA decreases significantly, as also illustrated in Fig.~\ref{fig:f_s_FAMA_W}. Also, the results show that given the same number of interfering BSs, MRT-$s$-FAMA with $K = 30$ significantly outperforms MRT-$f$-FAMA with $K = 25$. This demonstrates that with a large number of ports $K$, MRT-$s$-FAMA can be as effective as MRT-$f$-FAMA. But when the number of interfering BSs increases, given the same number of ports, the performance gap between MRT-$f$-FAMA and MRT-$s$-FAMA still increases.

\begin{figure}
\centering
\includegraphics[width=1.0\linewidth]{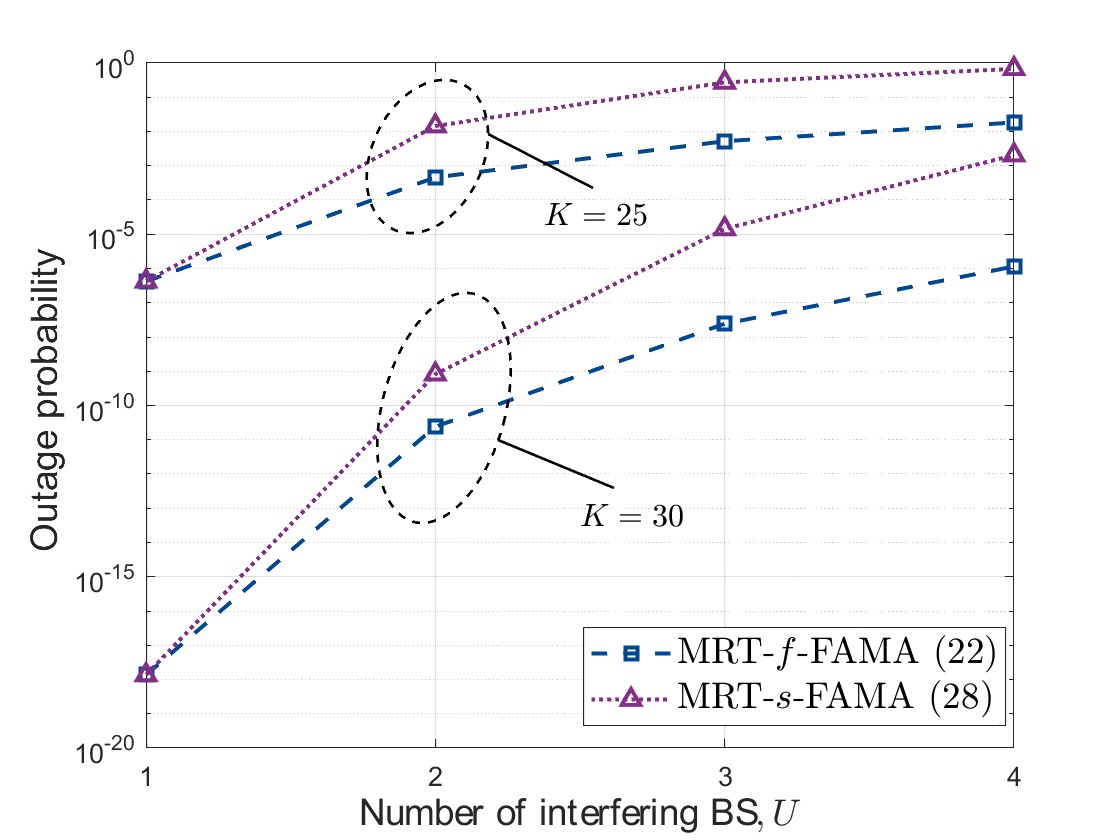}
\caption{Outage probability against the number of interfering BSs $U$, given different number of ports $K$, when the SIR threshold $\gamma^{\rm [f]}_{\rm th}= \gamma^{\rm [s]}_{\rm th} = 3$.}\label{fig:OP_U_threshold_3}
\end{figure} 

\section{Conclusion}\label{sec:conclude}
This paper considered a cell-free FAMA network in which the BS uses MRT precoding and each user adopts a FAMA approach to combat fading and interference in the downlink. This work is important in understanding the synergy between precoding and FAMA. Our emphasis was on the interference-limited environments where we derived the outage probability expressions for both MRT-$f$-FAMA and MRT-$s$-FAMA. Our results showed that while MRT-$f$-FAMA outperforms MRT-$s$-FAMA with fewer BS antennas, MRT-$s$-FAMA excels when the number of transmit antennas per BS is large. This result suggests that $s$-FAMA can be an attractive solution to combine with MRT precoding at the BS. Therefore, incorporating FAS on the user side can alleviate the CSI burden on BSs, thereby improving the scalability of cell-free networks. Further research could explore the application of FAS with alternative precoding methods, such as zero-forcing when more CSI can be afforded. Additionally, future work might investigate the performance of FAS under imperfect precoding scenarios.

\appendix
\subsection{Proof of Theorem \ref{g_k_Nakagami_RV}}\label{appendix:Proof_g_k_Nakagami_RV}
Note that $\sum_{n=1}^{N}\left |  h_{n,k} \right |^2$ is a sum over a set of uncorrelated squared Rayleigh random variables, which follows a Gamma distribution, with a PDF given by
\begin{align}
f_{\sum_{n=1}^{N}\left |  h_{n,k} \right |^2}(\tau) = {\rm Gamma}(\tau; N,\sigma^2).
\end{align}
According to \eqref{MRT_signal_port_k}, $\left|g_{k}\right|^2 = r_{0}^{-\alpha}\sum_{n=1}^{N}\left |  h_{n,k} \right |^2$. Utilizing the scaling property of the Gamma distribution, we obtain
\begin{align}
f_{\left|g_{k}\right|^2}(\tau)  = {\rm Gamma}(\tau; N,r_{0}^{-\alpha}\sigma^2),
\end{align}
and the correlation parameter between the channels $\left|g_{0}\right|^2$ and $\left|g_{k}\right|^2$ can be calculated as
\begin{align}\label{correlation_calculation_signal}
\rho_{\left|g_{0}\right|^2,\left|g_{k}\right|^2} = \frac{\mathrm{E}[\left|g_{0}\right|^2 \left|g_{k}\right|^2] - \mathrm{E}[\left|g_{0}\right|^2]\mathrm{E}[\left|g_{k}\right|^2]}{\sqrt{\mathrm{Var}[\left|g_{0}\right|^2] \mathrm{Var}[\left|g_{k}\right|^2]}},
\end{align}
where 
\begin{equation}
\left\{\begin{aligned}
{\rm{E}}[{\left| {{g_0}} \right|^2}{\left| {{g_k}} \right|^2}] &= r_0^{ - 2\alpha }N{\mu ^2}{\sigma ^4} + r_0^{ - 2\alpha }{N^2}{\sigma ^4},\\
{\rm{E}}[{\left| {{g_0}} \right|^2}] = {\rm{E}}[{\left| {{g_k}} \right|^2}] &= r_0^{ - 2\alpha }N{\sigma ^2},\\
{\rm{Var}}[{\left| {{g_0}} \right|^2}] = {\rm{Var}}[{\left| {{g_k}} \right|^2}] &= r_0^{ - 2\alpha }N{\sigma ^4}.
\end{aligned}\right.
\end{equation}
Consequently, we have $\rho_{\left|g_{0}\right|^2,\left|g_{k}\right|^2} = \mu^2$. The squared root of the Gamma distribution follows a Nakagami distribution, where the shape parameter $\omega = N$ remains unchanged and the spread parameter $\iota = r_{0}^{-\alpha}N\sigma^2$. As for the Nakagami distribution, it is required that its shape parameter $\omega \geq 0.5$ and spread parameter $\iota>0$. It is evident that $\omega = N \geq 1$. Additionally, $\iota = N\sigma^2\sum_{i=1}^{U}r_{i}^{-\alpha} >0$. Following \cite[(10)]{rician_rv}, the correlation parameter between $\left|g_{0}\right|$ and $\left|g_{k}\right|$ is the same as \eqref{correlation_calculation_signal}, given by  $\rho_{\left|g_{0}\right|,\left|g_{k}\right|} = \mu^2$. This ends the proof.

\subsection{Proof of Theorem \ref{Theorem_Out_f_FAMA}}\label{Proof_Out_f_FAMA}
The outage probability of MRT-$f$-FAMA can be expressed as \eqref{Appendix_A_1}, as shown at the top of next page, where $(a)$ comes directly from \eqref{OUT_f_FAMA}, $(b)$ is based on the law of total probability, where $F_{\left| g_{1}\right|, \dots, \left| g_{K}\right|}(\Theta_{\rm f} t_{1}, \dots, \Theta_{\rm f} t_{K}|t_{1},\dots, t_{K})$ is the joint CDF of $\left| g_{1}\right|, \dots, \left| g_{K}\right|$ conditioned on $t_{1},\dots,t_{K}$ given in Theorem \ref{joint_CDF_signal_channel}, $f_{\left| g_{1}^{\rm I} \right|, \dots, \left| g_{K}^{\rm I} \right|}(t_{1}, \dots, t_{K})$ is the joint PDF of $\left| g_{1}^{\rm I} \right|, \dots, \left| g_{K}^{\rm I} \right|$ given in \cite[Theorem 1]{Fluid_antenna_system}, $(c)$ has substituted these results, and finally, based on Tonelli's theorem \cite{Tonelli_theorem}, $(d)$ has changed the order of integration.  
\begin{figure*}[]
\centering
\begin{align}\label{Appendix_A_1}
{\mathcal O}_{{\rm MRT}\text{-}f\text{-}{\rm FAMA}} 
&\overset{(a)}{=} \mathbb{P}\left\{ \left| g_{1}\right| < \Theta_{\rm f} \left| g_{1}^{\rm I} \right|, \dots, \left| g_{K}\right| < \Theta_{\rm f} \left| g_{K}^{\rm I} \right| \right\}  \nonumber \\
&\overset{(b)}{=}\int_{t_K = 0}^{\infty} \cdots \int_{t_1 = 0}^{\infty} F_{\left| g_{1}\right|, \dots, \left| g_{K}\right|}(\Theta_{\rm f} t_{1}, \dots, \Theta_{\rm f} t_{K}|t_{1},\dots,t_{K}) \times f_{\left| g_{1}^{\rm I} \right|, \dots, \left| g_{K}^{\rm I} \right|}(t_{1}, \dots, t_{K})dt_{1}\dots dt_{K}\nonumber \\
& \overset{(c)}{=} \int_{t_K = 0}^{\infty} \cdots \int_{t_1 = 0}^{\infty} \frac{2\omega^{\omega}}{\Gamma(\omega)\iota^{\omega}} \int_{t = 0}^{\infty}t^{2\omega-1}e^{-\frac{\omega t^2}{\iota}}\prod_{k=1}^{K}\left[1- Q_{\omega}\left(\sqrt{\frac{2\omega\mu^2t^2}{\iota(1-\mu^2)}}, \sqrt{\frac{2\omega{\Theta^2_{\rm f}}t_{k}^2}{\iota(1-\mu^2)}}\right)\right]dt\nonumber \\
&\quad\quad\quad\quad\quad \times \int_{t_{0}=0}^{\infty}\frac{2t_{0}}{\sigma_{\rm I}^2}e^{-\frac{t_{0}^2}{\sigma_{\rm I}^2}}\prod_{k=1}^{K}\frac{2t_{k}}{\sigma_{\rm I}^2(1-\mu^2)}e^{-\frac{t_{k}^2+\mu^2t_{0}^2}{\sigma_{\rm I}^2(1-\mu^2)}}I_{0}\left(\frac{2\mu t_{0}t_{k}}{\sigma_{\rm I}^2(1-\mu^2)}\right)dt_{0}dt_{1}\dots dt_{K}
\nonumber \\
&\overset{(d)}{=}\frac{2\omega^{\omega}}{\Gamma(\omega)\iota^{\omega}}\int_{t_0 = 0}^{\infty}\frac{2t_{0}}{\sigma_{\rm I}^2}e^{-\frac{t_{0}^2}{\sigma_{\rm I}^2}}\int_{t = 0}^{\infty}t^{2\omega-1}e^{-\frac{\omega t^2}{\iota}} \left(\frac{e^{-\frac{\mu^2t_{0}^2}{\sigma_{\rm I}^2(1-\mu^2)}}}{\sigma_{\rm I}^2(1-\mu^2)}\right)^K \nonumber \\
&\quad\times\prod_{k=1}^{K}\left\{\int_{t_k = 0}^{\infty} 2t_{k}e^{-\frac{t_{k}^2}{\sigma_{\rm I}^2(1-\mu^2)}}\left[ 1- Q_{\omega}\left(\sqrt{\frac{2\omega\mu^2t^2}{\iota(1-\mu^2)}}, \sqrt{\frac{2\omega{\Theta^2_{\rm f}}t_{k}^2}{\iota(1-\mu^2)}}\right) \right]I_{0}\left(\frac{2\mu t_{0}t_{k}}{\sigma_{\rm I}^2(1-\mu^2)}\right)dt_{k}
\right\}dt dt_{0}
\end{align}
\hrulefill
\end{figure*}

We first evaluate the integration over $t_{k}$. According to \cite[(2)]{negative_MArcum_Q}, we have 
\begin{multline}\label{negative_Marcum_Q}
1- Q_{\omega}\left(\sqrt{\frac{2\omega\mu^2t^2}{\iota(1-\mu^2)}}, \sqrt{\frac{2\omega{\Theta^2_{\rm f}}t_{k}^2}{\iota(1-\mu^2)}}\right)\\ 
=Q_{1-\omega}\left(\sqrt{\frac{2\omega{\Theta^2_{\rm f}}t_{k}^2}{\iota(1-\mu^2)}}, \sqrt{\frac{2\omega\mu^2t^2}{\iota(1-\mu^2)}}\right).
\end{multline}
Therefore, the integration over $t_{k}$ can be expressed as
\begin{multline}\label{integration_Marum_Q_Bessel}
\int_{t_k = 0}^{\infty}2t_{k}e^{-\frac{t_{k}^2}{\sigma_{\rm I}^2(1-\mu^2)}}Q_{1-\omega}\left(\sqrt{\frac{2\omega{\Theta^2_{\rm f}}t_{k}^2}{\iota(1-\mu^2)}}, \sqrt{\frac{2\omega\mu^2t^2}{\iota(1-\mu^2)}}\right)\\
\times I_{0}\left(\frac{2\mu t_{0}t_{k}}{\sigma_{\rm I}^2(1-\mu^2)}\right)dt_{k}.
\end{multline}
With \cite[Proposition 1]{integration_marcum_Q_bessel}, we obtain the integral-form of \eqref{integration_Marum_Q_Bessel}, given in \eqref{closed_form_integration_Marcum_Bessel_1}, as shown at the top of next page.
\begin{figure*}[]
\begin{align}\label{closed_form_integration_Marcum_Bessel_1}
& \int_{t_k = 0}^{\infty}2t_{k}e^{-\frac{t_{k}^2}{\sigma_{\rm I}^2(1-\mu^2)}}Q_{1-\omega}\left(\sqrt{\frac{2\omega{\Theta^2_{\rm f}}t_{k}^2}{\iota(1-\mu^2)}}, \sqrt{\frac{2\omega\mu^2t^2}{\iota(1-\mu^2)}}\right) I_{0}\left(\frac{2\mu t_{0}t_{k}}{\sigma_{\rm I}^2(1-\mu^2)}\right)dt_{k} \nonumber \\
&= \sigma_{\rm I}^2(1-\mu^2)e^{\frac{\mu^2 t_{0}^2}{\sigma_{\rm I}^2(1-\mu^2)}}Q_{1}\left(\sqrt{\frac{2\omega{\Theta^2_{\rm f}}\mu^2t_{0}^2}{(1-\mu^2)(\iota +\omega{\Theta^2_{\rm f}}\sigma_{\rm I}^2)}},\sqrt{\frac{2\omega\mu^2t^2}{(1-\mu^2)(\iota+\omega{\Theta^2_{\rm f}}\sigma_{\rm I}^2)}} \right)-\frac{\sigma_{\rm I}^2(1-\mu^2)}{\mu t_{0}}\left(\frac{\mu\iota t_{0}}{\iota+\omega{\Theta^2_{\rm f}}\sigma_{\rm I}^2} \right)^a  \nonumber \\
& \quad\quad \times\left( \frac{\mu t}{\Theta_{\rm f}} \right)^{1-\omega}
e^{\frac{\iota\mu^2t_{0}^2}{\sigma_{\rm I}^2(1-\mu^2)(\iota+\omega\Theta^2_{\rm f}\sigma_{\rm I}^2)}}
e^{-\frac{\omega\mu^2t^2}{(1-\mu^2)(\iota+\omega\Theta^2_{\rm f}\sigma_{\rm I}^2)}} \sum_{q=0}^{a-1}\sum_{p=0}^{a-q-1}  \left(\frac{\sigma_{\rm I}^2 t}{t_{0}} \sqrt{\frac{\omega (1-\mu^2)}{2\iota }} \right)^{q+p} \left(\frac{\iota+\omega{\Theta^2_{\rm f}}\sigma_{\rm I}^2}{\Theta_{\rm f}\sigma_{\rm I}^2}\sqrt{\frac{2}{\omega \iota(1-\mu^2)}} \right)^q \nonumber \\
&\quad\quad\quad\quad\quad\quad\quad\quad\quad \times  \left(\sqrt{\frac{2\omega{\Theta^2_{\rm f}}}{\iota(1-\mu^2)}} \right)^{p}\frac{(a-q-p)_{p}}{p!}I_{1-\omega+q+p}\left( \frac{2\omega\Theta_{\rm f}\mu^2tt_{0}}{(1-\mu^2)(\iota+\omega{\Theta^2_{\rm f}}\sigma_{\rm I}^2)}\right)
\end{align}
\hrulefill
\end{figure*}

Finally, submitting the result in \eqref{closed_form_integration_Marcum_Bessel_1} into \eqref{Appendix_A_1}, the outage probability of MRT-$f$-FAMA can be expressed as \eqref{Outage_f_FAMA_Closedform}.

\subsection{Proof of Corollary \ref{Corollary_Out_f_FAMA}}\label{Proof_Corollary_Out_f_FAMA}
According to \cite[(9.6.19)]{intergation_form_bessel_function}, we have
\begin{align}\label{integration_form_bessel}
I_{1-\omega+q+p}(\tau) = \frac{1}{\pi}\int_{0}^{\pi} \cos{(\theta(1-\omega+q+p))}e^{\tau\cos{\theta}}d\theta.
\end{align}
Then by combining the exponential terms in \eqref{Outage_f_FAMA_Closedform} with \eqref{integration_form_bessel}, we have
\begin{align}\label{Modified_Bessel_exp_f_FAMA}
&e^{-\frac{\omega\mu^2\Theta_{\rm f}^2t_{0}^2+\omega\mu^2t^2}{(1-\mu^2)(\iota+\omega\Theta_{\rm f}^2\sigma_{\rm I}^2)}}I_{1-\omega+q+p}\left( \frac{2\omega\Theta_{\rm f}\mu^2tt_{0}}{(1-\mu^2)(\iota+\omega\Theta_{\rm f}^2\sigma_{\rm I}^2)}\right) \nonumber \\
= &\,e^{-\frac{(\sqrt{\omega}\mu\Theta_{\rm f}t_{0}-\sqrt{\omega}\mu t)^2}{(1-\mu^2)(\iota+\omega\Theta_{\rm f}^2\sigma_{\rm I}^2)}}\nonumber\\
&\times\frac{1}{\pi}\int_{0}^{\pi} \cos{(\theta(1-\omega+q+p))}e^{- \frac{4\omega\Theta_{\rm f}\mu^2tt_{0}}{(1-\mu^2)(\iota+\omega\Theta_{\rm f}^2\sigma_{\rm I}^2)}\sin^2{\frac{\theta}{2}}}d\theta. 
\end{align}
Submitting \eqref{Modified_Bessel_exp_f_FAMA} into \eqref{Outage_f_FAMA_Closedform} finally gives \eqref{Outage_f_FAMA_No_Besel}.

\subsection{Proof of Theorem \ref{g_k_s_Nakagami_RV}}\label{appendix:Proof_g_k_s_Nakagami_RV}
According to \eqref{interference_channel}, the channel $h_{k}^{(i)}$ follows a Rayleigh distribution and therefore, $\left | h_{k}^{(i)} \right |^2$ follows a Gamma distribution. Due to the scaling property of Gamma distribution, the PDF of $\left |r_{i}^{-\frac{\alpha}{2}} h_{k}^{(i)} \right |^2$ is given by
\begin{align}
f_{\left |r_{i}^{-\frac{\alpha}{2}} h_{k}^{(i)} \right |^2} (\tau) = {\rm{Gamma}}(\tau;1,r_{i}^{-\alpha}\sigma^2).
\end{align}
Therefore, $\left|g_{k}^{\rm [s]} \right|^2 = \sum_{i=1}^{U}\left |r_{i}^{-\frac{\alpha}{2}} h_{k}^{(i)} \right |^2$ is a sum over a set of uncorrelated Gamma random variable. Following \cite[(6)]{2000_sum_Gamma}, \cite[(2)]{2007_sum_Gamma} and \cite[Proposition 8]{2011_sum_Gamma}, $\left|g_{k}^{\rm [s]} \right|^2$ can be approximated by a Gamma distribution with the PDF given by
\begin{align}
f_{\left|g_{k}^{\rm [s]} \right|^2}(\tau) = {\rm Gamma}(\tau;\Omega, \nu),
\end{align}
where $\Omega =  \frac{(\sum_{i=1}^{U}r_{i}^{-\alpha})^2}{\sum_{i=1}^{U}r_{i}^{-2\alpha}}$ and $\nu = \frac{\sigma^2\sum_{i=1}^{U}{r_{i}^{-2\alpha}}}{\sum_{i=1}^{U}r_{i}^{-\alpha}}$. The correlation parameter between $\left|g_{0}^{[\rm {s}]}\right|^2$ and $\left|g_{k}^{[\rm {s}]}\right|^2$ is calculated as
\begin{align}\label{correlation_calculation}
\rho_{\left|g_{0}^{[\rm {s}]}\right|^2,\left|g_{k}^{[\rm {s}]}\right|^2} = \frac{\mathrm{E}\left[\left|g_{0}^{[\rm {s}]}\right|^2 \left|g_{k}^{[\rm {s}]}\right|^2\right] - \mathrm{E}\left[\left|g_{0}^{[\rm {s}]}\right|^2\right]\mathrm{E}\left[\left|g_{k}^{[\rm {s}]}\right|^2\right]}
{\sqrt{\mathrm{Var}\left[\left|g_{0}^{[\rm {s}]}\right|^2\right] \mathrm{Var}\left[\left|g_{k}^{[\rm {s}]}\right|^2\right]}},
\end{align}
where 
\begin{equation}
\left\{\begin{aligned}
{\rm{E}}\left[{\left| {{g_0}^{[\rm {s}]}} \right|^2}{\left| {{g_k}^{[\rm {s}]}} \right|^2}\right] &= N\mu^2\sigma^4\sum_{i=1}^{U}r_{i}^{-2\alpha}\\ 
&\quad +N^2\sigma^4\left(\sum_{i=1}^{U}r_{i}^{-\alpha}\right)^2,\\
{\rm{E}}\left[{\left| {{g_0}^{[\rm {s}]}} \right|^2}\right] &= {\rm{E}}\left[{\left| {{g_k}^{[\rm {s}]}} \right|^2}\right] = N\sigma^2\sum_{i=1}^{U}r_{i}^{-\alpha},\\
{\rm{Var}}\left[{\left| {{g_0}^{[\rm {s}]}} \right|^2}\right] &= {\rm{Var}}\left[{\left| {{g_k}^{[\rm {s}]}} \right|^2}\right] =N\sigma^4\sum_{i=1}^{U}r_{i}^{-2\alpha}.
\end{aligned}\right.
\end{equation}
Therefore, we have $\rho_{\left|g_{0}^{[\rm {s}]}\right|^2,\left|g_{k}^{[\rm {s}]}\right|^2} = \mu^2$. The squared root of the Gamma distribution follows a Nakagami distribution, where the shape parameter remains unchanged and the spread parameter is given by $\phi = \nu\Omega = \sigma^2\sum_{i=1}^{U}r_{i}^{-\alpha}$. Following \cite[(10)] {rician_rv}, the correlation parameter between $\left|g_{0}^{[\rm {s}]}\right|$ and $\left|g_{k}^{[\rm {s}]}\right|$ is the same as \eqref{correlation_calculation}, i.e., $\rho_{\left|g_{0}^{[\rm {s}]}\right|,\left|g_{k}^{[\rm {s}]}\right|} = \mu^2$.

\subsection{Proof of Theorem \ref{Theorem_Out_s_FAMA}}\label{Proof_Out_s_FAMA}
Similar to Appendix \ref{Proof_Out_f_FAMA}, the outage probability for MRT-$s$-FAMA can be expressed as \eqref{Appendix_B_1}, as shown at the top of next page, where $(a)$ comes from \eqref{OUT_s_FAMA}, $(b)$ is due to the law of total probability, where $F_{\left| g_{1}\right|, \dots, \left| g_{K}\right|}(\Theta_{\rm s} t_{1}, \dots, \Theta_{\rm s} t_{K}|t_{1},\dots, t_{K})$ is the CDF of $\left| g_{1}\right|, \dots, \left| g_{K}\right|$ conditioned on $t_{1},\dots,t_{K}$, given in Theorem \ref{joint_CDF_signal_channel}, and $f_{\left| g_{1}^{\rm [s]} \right|, \dots, \left| g_{K}^{\rm [s]} \right|}(t_{1}, \dots, t_{K})$ is the PDF of $\left| g_{1}^{\rm [s]} \right|, \dots, \left| g_{K}^{\rm [s]} \right|$ given in Theorem \ref{theorem_joint_PDF_Slow_Interference}. Finally, $(d)$ has changed the order of integration.
\begin{figure*}
\centering
\begin{align}\label{Appendix_B_1}
{\mathcal O}_{{\rm MRT}\text{-}s\text{-}{\rm FAMA}} 
&\overset{(a)}{=} \mathbb{P}\left\{ \left| g_{1}\right| < \Theta_{\rm [s]} \left| g_{1}^{\rm [s]} \right|, \dots, \left| g_{K}\right| < \Theta_{\rm [s]} \left| g_{K}^{\rm [s]} \right| \right\}  \nonumber \\
&\overset{(b)}{=}\int_{t_K = 0}^{\infty} \cdots \int_{t_1 = 0}^{\infty} F_{\left| g_{1}\right|, \dots, \left| g_{K}\right|}(\Theta_{\rm [s]} t_{1}, \dots, \Theta_{\rm [s]} t_{K}|t_{1},\dots,t_{K}) \times f_{\left| g_{1}^{\rm [s]} \right|, \dots, \left| g_{K}^{\rm [s]} \right|}(t_{1}, \dots, t_{K})dt_{1}\dots dt_{K}\nonumber \\
& \overset{(c)}{=} \int_{t_K = 0}^{\infty} \cdots \int_{t_1 = 0}^{\infty} \frac{2\omega^{\omega}}{\Gamma(\omega)\iota^{\omega}} \int_{t = 0}^{\infty}t^{2\omega-1}e^{-\frac{\omega t^2}{\iota}}\prod_{k=1}^{K}\left[1- Q_{\omega}\left(\sqrt{\frac{2\omega\mu^2t^2}{\iota(1-\mu^2)}}, \sqrt{\frac{2\omega\Theta_{\rm [s]}^2t_{k}^2}{\iota(1-\mu^2)}}\right)\right]dt\nonumber \\
&\quad\quad\quad \times \int_{t_{0}=0}^{\infty}\frac{2\Omega^{\Omega} t_{0}^{2\Omega-1}e^{-\frac{\Omega t_{0}^2}{\phi}}}{\Gamma(\Omega)\phi^{\Omega}}\prod_{k=1}^{K}\frac{\Omega t_{0}^{1-\Omega}e^{-\frac{\Omega\mu^2 t_{0}^2}{\phi(1-\mu^2)}}}{\phi (1-\mu^2)\mu^{\Omega-1}}2t_{k}^{\Omega}e^{-\frac{\Omega t_{k}^2}{\phi(1-\mu^2)}}I_{\Omega-1}\left(\frac{2\Omega\mu t_{0}t_{k}}{\phi(1-\mu^2)}\right)dt_{0}dt_{1}\dots dt_{K}\nonumber \\
&\overset{(d)}{=}\frac{2\omega^{\omega}}{\Gamma(\omega)\iota^{\omega}}\int_{t_0 = 0}^{\infty}\frac{2\Omega^{\Omega} t_{0}^{2\Omega-1}e^{-\frac{\Omega t_{0}^2}{\phi}}}{\Gamma(\Omega)\phi^{\Omega}} \int_{t = 0}^{\infty}t^{2\omega-1}e^{-\frac{\omega t^2}{\iota}} \prod_{k=1}^{K}\Bigg\{\frac{\Omega t_{0}^{1-\Omega}e^{-\frac{\Omega\mu^2 t_{0}^2}{\phi(1-\mu^2)}}}{\phi (1-\mu^2)\mu^{\Omega-1}}\nonumber \\
&\quad\quad\quad\times\int_{t_k = 0}^{\infty} 2t_{k}^{\Omega}e^{-\frac{\Omega t_{k}^2}{\phi(1-\mu^2)}}\left[1- Q_{\omega}\left(\sqrt{\frac{2\omega\mu^2t^2}{\iota(1-\mu^2)}}, \sqrt{\frac{2\omega\Theta_{\rm [s]}^2t_{k}^2}{\iota(1-\mu^2)}}\right)\right]I_{\Omega-1}\left(\frac{2\Omega\mu t_{0}t_{k}}{\phi(1-\mu^2)}\right)dt_{k}\Bigg\}dt dt_{0} 
\end{align}
\hrulefill
\end{figure*}

Based on \eqref{negative_Marcum_Q}, the integration over $t_{k}$ is expressed as
\begin{multline}\label{integration_Marcum_Bessel_2}
\int_{t_k = 0}^{\infty}2t_{k}^{\Omega}e^{-\frac{\Omega t_{k}^2}{\phi(1-\mu^2)}} Q_{1-\omega}\left(\sqrt{\frac{2\omega\Theta_{\rm s}^2t_{k}^2}{\iota(1-\mu^2)}}, \sqrt{\frac{2\omega\mu^2t^2}{\iota(1-\mu^2)}}\right)\\
\times I_{\Omega-1}\left(\frac{2\Omega\mu t_{0}t_{k}}{\phi(1-\mu^2)}\right)dt_{k}
\end{multline}
From \cite[Proposition 1]{integration_marcum_Q_bessel}, we have the integral-form expression of \eqref{integration_Marcum_Bessel_2}, given in \eqref{closed_form_integration_Marcum_Bessel_2}, as shown at the top of next page. 
\begin{figure*}
\begin{align}\label{closed_form_integration_Marcum_Bessel_2}
&  \int_{t_k = 0}^{\infty}2t_{k}^{\Omega}e^{-\frac{\Omega t_{k}^2}{\phi(1-\mu^2)}} Q_{1-\omega}\left(\sqrt{\frac{2\omega\Theta_{\rm s}^2t_{k}^2}{\iota(1-\mu^2)}}, \sqrt{\frac{2\omega\mu^2t^2}{\iota(1-\mu^2)}}\right) I_{\Omega-1}\left(\frac{2\Omega\mu t_{0}t_{k}}{\phi(1-\mu^2)}\right)dt_{k} \nonumber \\
&= \frac{\phi(1-\mu^2)}{\Omega}\mu^{\Omega-1}t_{0}^{\Omega-1}e^{\frac{\Omega\mu^2 t_{0}^2}{\phi(1-\mu^2)}}Q_{\Omega}\left( \sqrt{\frac{2\omega\Omega\mu^2\Theta_{\rm s}^2 t_{0}^2}{(1-\mu^2)(\iota\Omega+\omega\phi\Theta_{\rm s}^2)}}, \sqrt{\frac{2\omega\Omega\mu^2t^2}{(1-\mu^2)(\iota\Omega+\omega\phi\Theta_{\rm s}^2)}}\right)\nonumber \\
&\quad\quad\quad -\frac{\phi(1-\mu^2)}{\mu t_{0} \Omega} \left(\frac{\iota\mu\Omega t_{0} }{\iota\Omega+\omega\phi\Theta_{\rm s}^2}\right)^{b} \left( \frac{\mu t}{\Theta_{\rm s}} \right)^{1-\omega} e^{\frac{\iota \mu^2 \Omega^2t_{0}^2 - \phi \omega\Omega\mu^2  t^2}{\phi(1-\mu^2)(\iota\Omega+\omega\phi\Theta_{\rm s}^2)} }\sum_{q=0}^{b-1}\sum_{p=0}^{b-q-1}  \left(\frac{\phi t}{\Omega t_{0}} \sqrt{\frac{\omega (1-\mu^2)}{2\iota }} \right)^{q+p}  \nonumber \\
& \quad\quad\quad\quad\quad\quad \times \left(\frac{\iota\Omega+\omega\phi\Theta_{\rm s}^2}{\phi\Theta_{\rm s}}\sqrt{\frac{2}{ \iota\omega(1-\mu^2)}} \right)^q \left(\sqrt{\frac{2\omega\Theta_{\rm s}^2}{\iota(1-\mu^2)}} \right)^{p}\frac{(b-q-p)_{p}}{p!}I_{1-\omega+q+p}\left( \frac{2\omega\Omega\Theta_{\rm s}\mu^2 t t_{0}}{(1-\mu^2)(\iota\Omega+\omega\phi\Theta_{\rm s}^2)}\right)
\end{align}
\hrulefill
\end{figure*}

Finally, submitting \eqref{closed_form_integration_Marcum_Bessel_2} into \eqref{Appendix_B_1}, the outage probability of MRT-$s$-FAMA is expressed as \eqref{Outage_S_FAMA_Closedform}.

\subsection{Proof of Corollary \ref{Corollary_Out_s_FAMA}}\label{Proof_Corollary_Out_s_FAMA}
Similar to Appendix \ref{Proof_Corollary_Out_f_FAMA}, we first rewrite the modified Bessel function based on \eqref{integration_form_bessel}, which gives 
\begin{align}\label{Modified_Bessel_exp_S_FAMA}
&e^{-\frac{\omega\Omega\mu^2\Theta_{\rm s}^2t_{0}^2+\omega\Omega\mu^2  t^2}{(1-\mu^2)(\iota\Omega+\omega\phi\Theta_{\rm s}^2)} }I_{1-\omega+q+p}\left( \frac{2\omega\Omega\Theta_{\rm s}\mu^2 t t_{0}}{(1-\mu^2)(\iota\Omega+\omega\phi\Theta_{\rm s}^2)}\right) \nonumber \\
= &\,e^{-\frac{(\sqrt{\omega\Omega}\mu \Theta_{\rm s} t_{0}-\sqrt{\omega\Omega}\mu t)^2}{(1-\mu^2)(\iota\Omega+\omega\phi\Theta_{\rm s}^2)} }\nonumber\\
&\times\frac{1}{\pi}\int_{0}^{\pi} \cos{(\theta(1-\omega+q+p))}e^{- \frac{4\omega\Omega\Theta_{\rm s}\mu^2tt_{0}}{(1-\mu^2)(\iota\Omega+\omega\phi\Theta_{\rm s}^2)}\sin^2{\frac{\theta}{2}}}d\theta.
\end{align}
Submitting \eqref{Modified_Bessel_exp_S_FAMA} into \eqref{Outage_S_FAMA_Closedform} finally yields \eqref{Outage_S_FAMA_No_Besel}.
			
\bibliographystyle{ieeetr}

\end{document}